\documentclass[preprint,11pt]{IEEEtran}
\usepackage{setspace}
\onecolumn
\usepackage{amssymb,amsfonts,amsmath,amsthm,amscd,dsfont,mathrsfs}
\usepackage{graphicx,float,psfrag,epsfig}
\usepackage{wrapfig}
\usepackage{relsize}

\DeclareMathAlphabet{\mathpzc}{OT1}{pzc}{m}{it}

\usepackage{subfigure}
\usepackage{color}
\usepackage[]{algorithm}
\usepackage{algorithmic}
\newtheorem{theorem}{Theorem}

\newtheorem{lemma}{Lemma}

\newtheorem{definition}{Definition}
\newtheorem{note}{Note}

\newcommand{\1}{{\bf 1}} 
\newcommand{\E}{\mathsf{E}} 

\newcommand{\card}[1]           {\left| #1\right|}

\newcommand{\bea}{\begin{eqnarray}}
\newcommand{\eea}{\end{eqnarray}}
\newcommand{\beas}{\begin{eqnarray*}}
\newcommand{\eeas}{\end{eqnarray*}}

\begin{document}
\title{Successive Refinement with Decoder \\ Cooperation and its Channel Coding Duals}

\author{Himanshu Asnani${}^{\dagger}$, Haim Permuter${}^{*}$ and Tsachy Weissman${}^{\dagger}$ 
            \footnote{${}^{\dagger}$Department of Electrical Engineering, Stanford University}
            \footnote{${}^{*}$Department of Electrical Engineering, Ben Gurion University}}

\maketitle

\begin{abstract}
We study cooperation in multi terminal source coding models involving successive refinement. Specifically, we study the case of  a single encoder and two decoders, where the encoder provides a common description to both the decoders and a private description to only one of the decoders. The decoders cooperate via \textit{cribbing}, i.e.,  the decoder with access only to the common description is allowed to observe, in addition,  a deterministic function of the reconstruction symbols produced by the other. We characterize the fundamental performance limits in the respective settings of non-causal, strictly-causal and causal cribbing. We use a new coding scheme, referred to as \textit{Forward Encoding and Block Markov Decoding}, which is a variant of one  recently used by Cuff and Zhao for coordination via implicit communication.  Finally, we  use the insight gained to introduce and solve some dual channel coding scenarios involving Multiple Access Channels with cribbing.\end{abstract}

\section{Introduction}
\label{sec::intro}
Cooperation can dramatically boost the performance of a  network. The literature abounds with models for cooperation, when communication between nodes of a network is over a noisy channel. In multiple access channels,  the setting of \textit{cribbing} was introduced by Willems and Van der Muelen in \cite{Willems_Cribbing}, where one encoder obtains the channel input symbols of the other encoder (referred to as \textquotedblleft crib\textquotedblright) and uses it for coding over a multiple access channel (MAC). This was further generalized to deterministic function cribbing (where an encoder obtains a deterministic function of the channel input symbols of another encoder) and to cribbing with actions (where one encoder can  control the quality and availability of the \textquotedblleft crib\textquotedblright\ by taking cost constrained actions) by Permuter and Asnani in \cite{HaimHimanshuCribbing}. Cooperation can also be modeled as information exchange among the transmitters and receivers via rate limited links, generally referred to as \textit{conferencing} in the literature. Such a model was introduced in the context of the MAC by Willems  in \cite{Willems_Conferencing}, and subsequently studied by by Bross, Lapidoth and Wigger \cite{Gaussian_MAC_Conferencing}, Wiese et al. \cite{Compound_MAC_Weisse}, Simeone et al. \cite{Compound_MAC_Partial_Simeone},  and Maric, Yates and Kramer \cite{Compound_MAC_Conferencing_Maric}. Cooperation has also been modeled via \textit{conferencing}/\textit{cribbing} in cognitive interference channels, such as the settings in Bross, Steinberg and Tinguely \cite{Causal_Cognitive_Bross} and Prabhakaran and Vishwanath \cite{Pramod_Vinod_Transmitter_Cooperation}-\cite{Pramod_Vinod_Receiver_Cooperation}.  We refer to Ng and Goldsmith \cite{Capacity_Cooperation_Goldsmith} for a survey of various cooperation strategies and their  fundamental limits in wireless networks.
\begin{figure}[htbp]
\begin{center}
\scalebox{0.7}{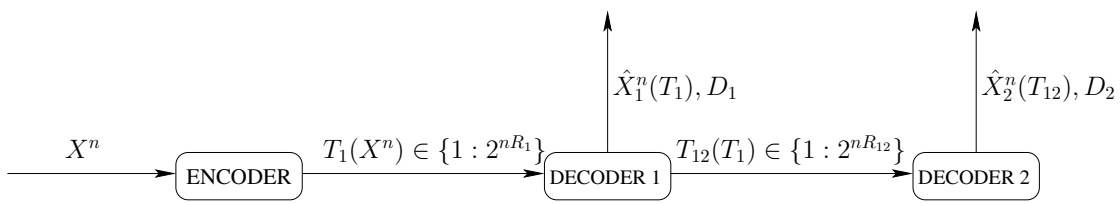}
\caption{Cascade source coding setup.}
\label{cascade}
\end{center}
\end{figure}
\begin{figure}[htbp]
\begin{center}
\scalebox{0.7}{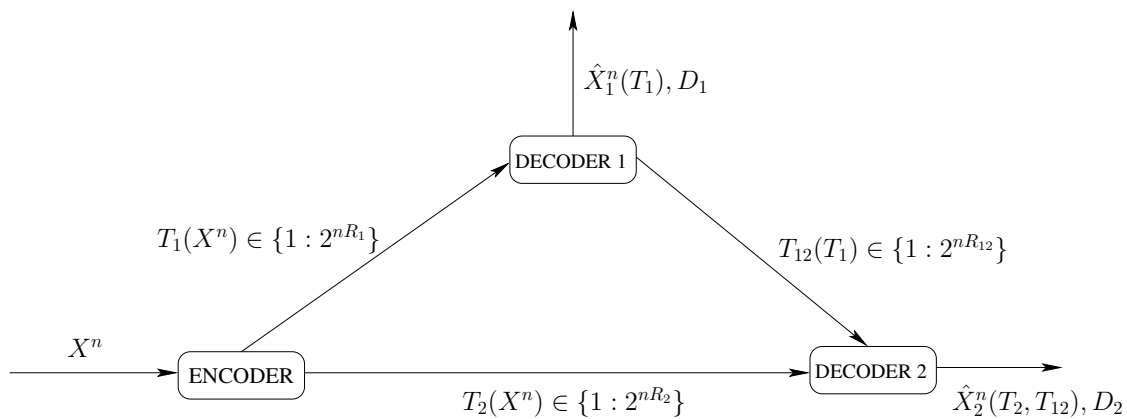}
\caption{Triangular source coding setup.}
\label{triangle}
\end{center}
\end{figure}
\par 
In multi terminal source coding, cooperation is generally modeled as a rate limited link such as in the cascade source coding  setting of Yamamoto \cite{Yamamoto}, Cuff, Su and El Gamal \cite{Han_Cuff_Cascade}, Permuter and Weissman \cite{HaimTsachyCascade}, Chia, Permuter and Weissman \cite{YeowHaimTsachy}, as well as the triangular source coding problems of Yamamoto \cite{Yamamoto_Triangle}, Chia, Permuter and Weissman \cite{YeowHaimTsachy}. In cascade source coding (Fig. \ref{cascade}), Decoder 1 sends a description ($T_{12}$) to Decoder 2, which does not receive any direct description from the encoder, while in triangular source coding (Fig. \ref{triangle}), Decoder 1 provides a description ($T_{12}$) to Decoder 2 in addition to the direct description ($T_
2$) from the encoder. 
\par
The contribution of this paper is to introduce new models of cooperation in multi terminal source coding, inspired by the \textit{cribbing} of Willems and Van der Muelen \cite{Willems_Cribbing} and by the implicit communication model of Cuff and Zhao \cite{Coordination_Implicit_Cuff_Zhao}. Specifically, we consider cooperation between decoders in a successive refinement setting (introduced in Equitz and Cover \cite{EquitzCover}). In successive refinement, a single encoder describes a common rate to both the decoders and a private rate to only one of the decoders. We generalize this model to accommodate  \textit{cooperation} among the decoders as follows : 

\begin{figure}[htbp]
\begin{center}
\scalebox{0.7}{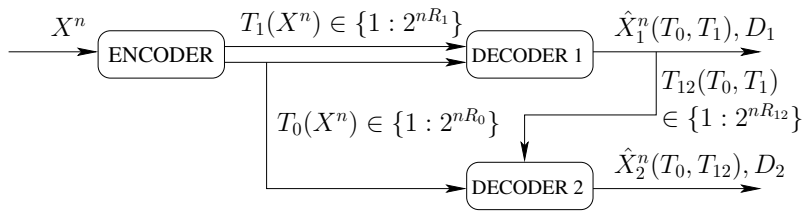}
\caption{Successive refinement, with decoders \textit{cooperating} via \textit{conferencing}.}
\label{sr_conferencing}
\end{center}
\end{figure} 
\begin{enumerate}
\item \textit{Cooperation via Conferencing} : One such cooperation model considered is that  shown in Fig. \ref{sr_conferencing}, where the encoder provides a common description ($T_0$) to both the decoders and a refined description ($T_1$) to Decoder 1, Decoder 1 cooperates with Decoder 2 by providing an additional description $(T_{12})$ which is the function of its own private description ($T_1$),  as well as the common description ($T_0$). This setting is inspired by the \textit{conferencing} problem in channel coding described earlier.  The region of achievable rates and distortions  for this problem is given by,
\bea
R_0+R_1&\ge& I(X;\hat{X}_1,\hat{X}_2)\label{eqpre1}\\
R_0+R_{12}&\ge& I(X;\hat{X}_2)\label{eqpre2},
\eea
for some joint probability distribution $P_{X,\hat{X}_1,\hat{X}_2}$ such that $\E[d_i(X_i,\hat{X}_i)]\le D_i$, for $i=1,2$, where $d_i$ refers to the distortion function and $D_i$ are the distortion constraints, as is formally explained in Section \ref{sec::problem}.  The direct part of this characterization, namely that this region is achievable, follows standard arguments that generalize those used in the original successive refinement problem \cite{EquitzCover} (cf. Appendix \ref{sec::sr_conferencing}). 
 \begin{figure}[htbp]
\begin{center}
\scalebox{0.7}{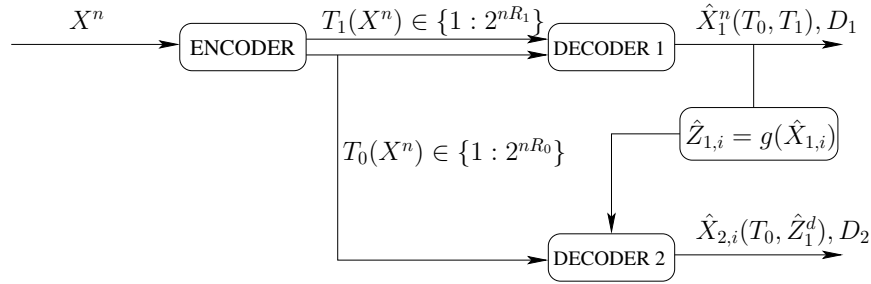}
\caption{Successive refinement, with decoders \textit{cooperating} via \textit{cribbing}. $d=n$, $d=i-1$ and $d=i$ respectively correspond to non-causal, strictly-causal and causal cribbing.}
\label{sr_cribbing}
\end{center}
\end{figure}
\item \textit{Cooperation via Cribbing} : 
The main setting analyzed in this paper is shown in Fig. \ref{sr_cribbing}. A single encoder describes a common message $T_0$ to both decoders and a refined message $T_1$ to only Decoder 1. Instead of cooperating via a rate limited link, as in Fig. \ref{sr_conferencing}, Decoder 2 \textquotedblleft cribs\textquotedblright\ (in the spirit of Willems and Van der Muelen \cite{Willems_Cribbing}) a deterministic function $g$ of the reconstruction symbols of Decoder 1, non-causally, strictly-causally, or causally.  Note a trivial $g$ function corresponds to the original successive refinement setting characterized in Equitz and Cover \cite{EquitzCover}. The goal is to find the optimal encoding and decoding strategy and to characterize the optimal encoding rate region which is defined as the set of achievable rate tuples $(R_0,R_1)$ such that the distortion constraints are satisfied at both the decoders. Cuff and Zhao \cite{Coordination_Implicit_Cuff_Zhao}, considered the problem of characterizing the coordination region (non-causal, strictly causal and causal coordination) in our setting of Fig. \ref{sr_cribbing}, for a specific function, $g$, such that $g(\hat{X}_1)=\hat{X}_1$ and for a specific rate tuple $(R_0,R_1)=(0,\infty)$, that is  Decoder 1 has access to the source sequence $X^n$ while Decoder 2 uses the reconstruction symbols of Decoder 1 (non-causally, strictly-causally or causally) to estimate the source. We use a new source coding scheme which we refer to as  \textit{Forward Encoding and Block Markov Decoding}, and show that it achieves the optimal rate region for strictly causal and causal cribbing. It draws on the achievability ideas (for causal coordination) introduced in Cuff and Zhao \cite{Coordination_Implicit_Cuff_Zhao}. This scheme operates in blocks, where in the current block, the encoder encodes for the source sequence of the future block, (hence the name \textit{Forward Encoding}) and the decoders rely on the decodability in the previous block to decode in the current block (hence the name \textit{Block Markov Decoding}). More details about this scheme are deferred to Section \ref{sec::sr_cribbing}.  
\end{enumerate}
\par
The general motivation for our work is an attempt to understand fundamental limits in source coding scenarios involving the availability of side information in the form of a lossily compressed version of the source. This is a departure from the standard and well studied models where side information is merely a correlated ``noisy version'' of the source, and is challenging because the effective `channel' from source to side information is now induced by a compression scheme. Thus, rather than dictated by nature, the side information is now another degree of freedom in the design. There is no shortage of practical scenarios that motivate our models. 
\par
One such scenario may arise in the context of video coding, as considered by Aaron, Varodayan and Girod  in (\cite{Aaron06wyner-zivresidual}). Consider two consecutive frames in a video file, denoted by Frame 1 and Frame 2, respectively. The video encoder starts by encoding Frame 1, and then it encodes the difference between Frame 1 and Frame 2. Decoder 1 represents  decoding of Frame 1, while Decoder 2 uses the knowledge of decoded Frame 1 (via cribbing) to estimate the next frame, Frame 2.  
\par
Our problem setting is equally natural for capturing noncooperation as it is for capturing cooperation, by requiring the relevant distortions to be bounded from below rather than above (which, in turn, can be converted to our standard form of an upper bound on the distortion by changing the sign of the distortion criterion). For instance, Decoder 1 can represent an end-user with refined information (common and private rate) about a secret document, the source in our problem, while Decoder 2 has a crude information about the document (via the  common rate). Decoder 1 is required to publicly announce an approximate version of the document, but due to privacy issues would like to remain somewhat cryptic about the source (as measured in terms of  distortion with respect to the source) while also helping (via conferencing or cribbing) Decoder 2 to better estimate the source. For example, Decoder 1 can represent a Government agency required by law to publicly reveal features of the data, while on the other hand there are agents who make use of this publicly announced information, along with crude information about the source that they too, not only the government, are allowed to access, to decipher  or get a good estimate of  the classified information (the source). 
\par
The contribution of this paper is two-fold. First, we introduce new models of decoder cooperation in source coding problems such as successive refinement, where decoders cooperate via cribbing, and we characterize the fundamental limits on performance for these problems using new classes of schemes for the achievability part. Second, we leverage the insights gained from these problems to introduce and solve a new class of channel coding scenarios that are dual to the source coding ones. Specifically, we consider the MAC with cribbing and a common message, where there are two encoders who want to communicate messages over the MAC, one has access to its own private message, there is a common message between the two encoders, and the encoders cooperate via cribbing (non-causally, strictly causally or causally). 
\par
The paper is organized as follows. Section \ref{sec::problem} gives a formal description of the problem and the main results. Section \ref{sec::sr_cribbing} presents achievability and converses, with non-causal, causal and strictly-causal cribbing. Some special cases of our setting and numerical examples, are studied in Section \ref{sec::special-case}. Channel coding duals are  considered in Section \ref{sec::duality}.  Finally, the paper is concluded in Section \ref{sec::conclusion}.

\section{Problem Definitions and Main Results}
\label{sec::problem}
We begin by explaining the notation to be used throughout this paper.
Let upper case, lower case, and calligraphic letters denote, respectively, random
variables, specific or deterministic values which random variables may assume,
and
their alphabets. For two jointly distributed random variables, $X$ and $Y$, let
$P_X$, $P_{XY}$ and $P_{X|Y}$ respectively denote the marginal of $X$, joint
distribution of $(X,Y)$ and conditional distribution of 
$X$ given $Y$. $X_{m}^{n}$ is a shorthand for the $n-m+1$ tuple
$\{X_m,X_{m+1},\cdots,X_{n-1},X_n\}$. We impose the assumption of finiteness of
cardinality on all alphabets, unless otherwise indicated.
\par
In this section we formally define the problem considered in this paper (cf. Fig. \ref{sr_cribbing}). The source sequence ${X_i\in \mathcal{X}, i = 1, 2,...}$ is drawn i.i.d. $\sim P_X$. Let $\hat{\mathcal{X}}_1$ and $\hat{\mathcal{X}}_2$ denote the reconstruction alphabets, and $d_i : \mathcal{X}\times\hat{\mathcal{X}}_i \rightarrow [0,\infty)$, for $ i = 1, 2$ denote single letter distortion measures. Distortion between sequences is defined in the usual way, 
\bea
d_i(x^n,\hat{x}_i^n)=\frac{1}{n}\sum_{j=1}^n d_i(x_j,\hat{x}_{i,j}), \mbox{ for }i=1,2.
\eea
\begin{definition}
\label{definition1}
A ($2^{nR_0},2^{nR_1},n$) rate-distortion  code consists of the following,
\begin{enumerate}
\item Encoder,  $f_{0,n} : \mathcal{X}^n \rightarrow \{1, . . . , 2^{nR_0}\}$, $f_{1,n} : \mathcal{X}^n \rightarrow \{1, . . . , 2^{nR_1}\}$.
\item Decoder 1, $g_{1,n} : \{1, . . . , 2^{nR_0}\}\times\{1,...,2^{nR_1}\} \rightarrow \hat{\mathcal{X}}_1^n$.
\item Decoder 2 (depending on $d$ in Fig. \ref{sr_cribbing}, the decoder mapping changes as below), 
\bea
g^{nc}_{2,i} : \{1, . . . , 2^{nR_0}\}\times\hat{\mathcal{X}}_1^n& \rightarrow &\hat{\mathcal{X}}_{2}\mbox{\textit{\ \ \ \ non-causal cribbing, $d=n$}}\\
g^{sc}_{2,i} : \{1, . . . , 2^{nR_0}\}\times\hat{\mathcal{X}}_1^{i-1}& \rightarrow& \hat{\mathcal{X}}_{2}\mbox{\textit{\ \ \ \ strictly-causal cribbing, $d=i-1$}},\\
g^{c}_{2,i} : \{1, . . . , 2^{nR_0}\}\times\hat{\mathcal{X}}_1^i &\rightarrow &\hat{\mathcal{X}}_{2}\mbox{\textit{\ \ \ \ causal cribbing, $d=i$}}
\eea
$\forall\ i=1,...,n$.
\end{enumerate}
\end{definition}
\begin{definition}
\label{definition2}
A rate-distortion tuple $(R_0,R_1,D_1,D_2)$ is said to be achievable if $\forall \ \epsilon > 0$, $\exists\ n$ and ($2^{nR_0},2^{nR_1},n$) rate-distortion  code such that the expected distortion for decoders are bounded as,
\bea
\E\left[d_i(X^n_i,\hat{X}^n_{i})\right]\le D_i+\epsilon,\mbox{\ }i=1,2.
\eea
\end{definition}
\begin{definition}
\label{definition3}
The rate-distortion region $\mathcal{R}(D_1,D_2)$ is defined as the closure of the set to all \textit{achievable} rate-distortion tuples $(R_0,R_1,D_1,D_2)$.
\end{definition}
Our main results for this setting are presented in the Table \ref{table_sr_cribbing}. Note that in all the rate regions in the table, we use the notation $\{a\}^+$ for $\max(a,0)$, and we omit the distortion condition $\E[d_i(X_i,\hat{X}_{i}]\le D_i$, $i=1,2$ for the sake of brevity. These results will be derived later in Section \ref{sec::sr_cribbing}. As another contribution, in Section \ref{sec::duality}, we establish duality between the problem of successive refinement with cribbing decoders and communication over multiple access channels with cribbing encoders and a common message. We establish a complete duality between the settings (in a  sense that is detailed in Section \ref{sec::duality}) and rate regions of one can be obtained from those of the other by listed transformations.
\begin{table}[h!]
\begin{center}

\begin{tabular}{|l|c|c|c|c|c|}
\hline
&&\\
$\mathcal{R}(D_1,D_2)$ &Perfect Cribbing & Deterministic Function \\
 &$g(\hat{X}_1)=\hat{X}_1$&Cribbing\\
&&\\
\hline
&&\\
 Non-Causal&(Theorem 1)&(Theorem 2)\\
 $(d=n)$&$R_0+R_1\ge I(X;\hat{X}_1,\hat{X}_2)$&$R_0+R_1\ge I(X;\hat{X}_1,\hat{X}_2)$\\
 &$R_0\ge \{I(X;\hat{X}_1,\hat{X}_2)- H(\hat{X}_1)\}^+$&$R_0\ge \{I(X;\hat{Z}_1,\hat{X}_2)- H(\hat{Z}_1)\}^+$\\
 &\textit{(p.m.f.)} :  $P(X,\hat{X}_1,\hat{X}_2)$&\textit{(p.m.f.)} :  $P(X,\hat{X}_1,\hat{X}_2)\1_{\{\hat{Z}_1=g(\hat{X}_1)\}}$\\
&&\\
\hline
&&\\
  Strictly-Causal&(Theorem 3)&(Theorem 4)\\
 $(d=i-1)$&$R_0+R_1\ge I(X;\hat{X}_1,\hat{X}_2)$&$R_0+R_1\ge I(X;\hat{X}_1,\hat{X}_2)$\\
 &$R_0\ge \{I(X;\hat{X}_1,\hat{X}_2)- H(\hat{X}_1|\hat{X}_2)\}^+$&$R_0\ge \{I(X;\hat{Z}_1,\hat{X}_2)- H(\hat{Z}_1|\hat{X}_2)\}^+$\\
 &\textit{(p.m.f.)} : $P(X,\hat{X}_1,\hat{X}_2)$&\textit{(p.m.f.)} :  $P(X,\hat{X}_1,\hat{X}_2)\1_{\{\hat{Z}_1=g(\hat{X}_1)\}}$\\
&&\\
\hline

&&\\
  Causal&(Theorem 5)&(Theorem 6)\\
 $(d=i)$&$R_0+R_1\ge I(X;\hat{X}_1,U)$&$R_0+R_1\ge I(X;\hat{X}_1,U)$\\
 &$R_0\ge \{I(X;\hat{X}_1,U)- H(\hat{X}_1|U)\}^+$&$R_0\ge \{I(X;\hat{Z}_1,U)- H(\hat{Z}_1|U)\}^+$\\
 &\textit{(p.m.f.)} : $P(X,\hat{X}_1,U) \1_{\{\hat{X}_2=f(U)\}}$&\textit{(p.m.f.)} : $P(X,\hat{X}_1,U) \1_{\{\hat{Z}_1=g(\hat{X}_1),\hat{X}_2=f(\hat{X}_1,U)\}}$\\
  &$\card{\mathcal{U}}\le\card{\mathcal{X}}\card{\mathcal{X}_1}+4$&$\card{\mathcal{U}}\le\card{\mathcal{X}}\card{\mathcal{X}_1}+4$\\
 &&\\
\hline

\end{tabular}
\end{center}

\caption{Main Results of the Paper}
\label{table_sr_cribbing}
\end{table}

\begin{lemma}[Equivalence to Cascade Source Coding with Cribbing Decoders]
\label{lemma1}
The setup in Fig. \ref{sr_cribbing} is equivalent to a cascade source coding setup with cribbing decoders as in Fig. \ref{cascade_cribbing} in the following way : fix a distortion pair $(D_1,D_2)$ and let $\mathcal{R}(D_1,D_2)$ denote the rate region for the problem of successive refinement with cribbing with achievable rate pairs $(R_0,R_1)$. Let $\tilde{R}(D_1,D_2)$ denote the closure of rate pairs, $(R_0,R_0+R_1)$ and ${\mathcal{R}}_{cascade}(D_1,D_2)$ denote the rate region for the problem of cascade source coding with cribbing (closure of achievable rate pairs $(R_{12},R_1)$).  We then have the equivalence, $\tilde{\mathcal{R}}(D_1,D_2)=\mathcal{R}_{cascade}(D_1,D_2)$. 
\end{lemma}
\begin{proof}
Proof is similar to the proof of Theorem 3 in Vasudevan, Tian and Diggavi \cite{Vasudevan_Diggavi_Tian_Cascade}. We state it in Appendix \ref{appendix_cascade} for quick reference. 
\end{proof}
\begin{figure}[h!]
\begin{center}
\scalebox{0.7}{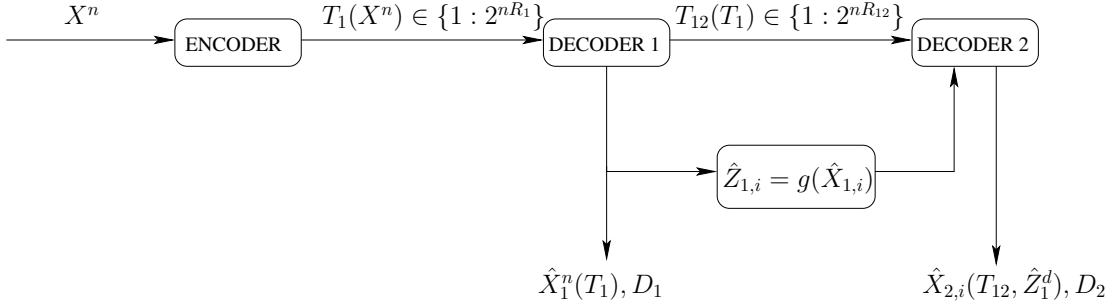}
\caption{Cascade source coding with \textit{cribbing} decoders,  $d=n$, $d=i-1$ and $d=i$ respectively correspond to non-causal, strictly-causal and causal cribbing.}
\label{cascade_cribbing}
\end{center}
\end{figure}
We use certain standard techniques such as Typical Average Lemma, Covering Lemma and Packing Lemma which are stated and established in \cite{GamalKim}. Herein, we state them for the sake of quick reference.  For typical sets we use the definition as in chapter 2 of \cite{GamalKim}. Henceforth, we omit the alphabets from the notation of typical set when it is clear from context, e.g. $\mathcal{T}^n_\epsilon(X,\hat{X}_2)$ is denoted by $\mathcal{T}^n_\epsilon$.
\begin{lemma}[Typical Average Lemma, Chapter 2, \cite{GamalKim}]
\label{typicalaverage}
Let $x^n \in T_\epsilon^n$. Then for any nonnegative function $g(x)$ on $\mathcal{X}$, 
\bea
(1-\epsilon)\E[g(X)]\le \frac{1}{n}\sum_{i=1}^n g(x_i)\le (1+\epsilon)\E[g(X)].
\eea
\end{lemma}
\begin{lemma}[Covering Lemma, Chapter 3, \cite{GamalKim}]
\label{covering}
Let $(U,X,\hat{X})\sim p(u,x,\hat{x})$. Let $(U^n,X^n)\sim p(u^n,x^n)$ be a pair of arbitrarily distributed random sequences such that $P\{(U^n,X^n) \in T_\epsilon^n\}\rightarrow 1$ as
$n\rightarrow\infty$ and let $\hat{X}^n(m),m \in \mathcal{A}$, where $\card{\mathcal{A}} \ge 2^{nR}$, be random sequences, conditionally independent of each other and of $X^n$ given $U^n$, each distributed according to $\prod_{i=1}^n p_{\hat{X}|U}(\hat{x}_i|u_i)$. Then, there exists $\delta(\epsilon)\rightarrow 0$ such that $P\{(U^n,X^n,\hat{X}^n(m))\notin T_{\epsilon}^n \ \forall \ m \in \mathcal{A}\}\rightarrow 0$ as $n\rightarrow \infty$ , if $R > I (X; \hat{X}|U) + \delta ( \epsilon )$.
\end{lemma}
\begin{lemma}[Packing Lemma, Chapter 3, \cite{GamalKim}]
\label{packing}
Let $(U,X,Y )\sim p(u,x,y)$. Let $(\tilde{U}^n,\tilde{Y}^n)\sim p(\tilde{u}^n,\tilde{y}^n)$ be a pair of arbitrarily distributed random sequences (not necessarily according to $\prod_{i=1}^n p_{U,Y}(\tilde{u}_i,\tilde{y}_i)$). Let $X^n(m), m \in \mathcal{A}$, where $\card{\mathcal{A}} \le  2^{nR}$, be random sequences, each distributed according to $\prod_{i=1}^n p_{\hat{X}|U}(\hat{x}_i|u_i)$. Assume that $X^n(m), m \in \mathcal{A}$, is pairwise conditionally independent of $\tilde{Y}^n$ given $\tilde{U}^n$, but is arbitrarily dependent on other $X^n(m)$ sequences. Then, there exists $\delta(\epsilon)\rightarrow 0$ such that $P\{\tilde{U}^n,X^n,\tilde{Y}^n(m))\in T_{\epsilon}^n \ \forall \ m \in \mathcal{A}\}\rightarrow 0$ as $n\rightarrow \infty$ , if $R < I (X;Y|U) + \delta ( \epsilon )$.
\end{lemma}

\section{Successive Refinement with Cribbing Decoders}
\label{sec::sr_cribbing}
In this section we analyze the main settings considered in this paper and derive rate regions.  In the various subsections to follow we will respectively study the problem of successive refinement with non-causal, strictly causal and causal cribbing. For clarity, in each subsection, we will first study the setting of \textquotedblleft perfect\textquotedblright\ cribbing where $\hat{Z}_{1,i}=g(\hat{X}_{1,i})=\hat{X}_{1,i}$ and then generalize it to cribbing with any deterministic function $g$.

\subsection{Non-causal Cribbing}
\label{subsec::sr_cribbing_nc}
\subsubsection{Perfect Cribbing}
\label{subsubsec::sr_cribbing_nc_perfect}
\begin{figure}[htbp]
\begin{center}
\scalebox{0.7}{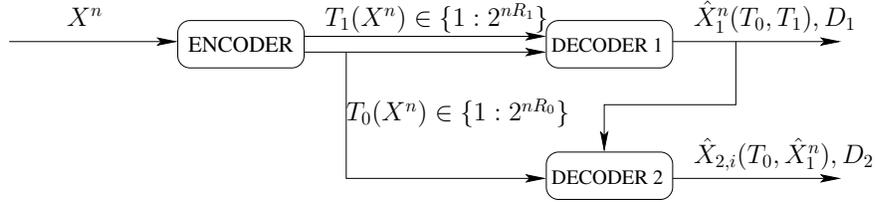}
\caption{Successive refinement, with decoders \textit{cooperating} via (perfect) \textit{non-causal cribbing}.}
\label{sr_cribbing_nc}
\end{center}
\end{figure}

\begin{theorem}
\label{theorem2}
 The rate region $\mathcal{R}(D_1,D_2)$ for the setting in Fig. \ref{sr_cribbing_nc} with perfect (non-causal) cribbing is given as the closure of the set of all the rate tuples $(R_0,R_1)$ such that,
\bea
R_0+R_1&\ge&I(X;\hat{X}_1,\hat{X}_2)\\
R_0&\ge& \{I(X;\hat{X}_1,\hat{X}_2)-H(\hat{X}_1)\}^+,
\eea
for some joint probability distribution $P_{X,\hat{X}_1,\hat{X_2}}$ such that 
$\E[d_i(X,\hat{X}_i)]\le D_i$, for $i=1,2$.
\end{theorem}
\begin{proof}
\par
\textit{Achievability} : \par
\underline{\textit{\textquotedblleft Double Binning\textquotedblright\ scheme}}
\par
Before delving into the details, we first provide a high level understanding of the achievability scheme.  Consider the simplified setup where $R_0=0$, that is only Decoder 1 has access to the description of the source, and Decoder 2 gets the reconstruction symbols of Decoder 1 (\textquotedblleft crib\textquotedblright). The intuition is to reveal a lossy description of source to the Decoder 2 through the \textquotedblleft crib\textquotedblright.  So we first generate $2^{nI(X;\hat{X}_2)}$ $\hat{X}_2^n$ codewords, and index them as $2^{nI(X;\hat{X}_2)}$ bins. In each bin, we generate a superimposed codebook of $2^{nI(X;\hat{X}_1|\hat{X}_2)}$ $\hat{X}_1^n$ codewords. Thus total rate of $R_1=I(X;\hat{X}_2)+I(X;\hat{X}_1|\hat{X}_2)=I(X;\hat{X}_1,\hat{X}_2)$ is needed to describe $\hat{X}_1^n$ to Decoder 1. Decoder 2 knows $\hat{X}_1^n$ via the crib, it then needs to infer the unique bin index which was sent, as then it would infer $\hat{X}_2^n$. The only issue to verify is that the $\hat{X}_1^n$ codeword known via cribbing should not lie in two bins. We upper bound the probability of occurrence of such an event by $2^{n(I(X;\hat{X}_1,\hat{X}_2)-H(\hat{X}_1))}$, as there are overall $2^{nI(X;\hat{X}_1,\hat{X}_2)}$ $\hat{X}_1^n$ codewords, and the probability that  a particular $\hat{X}_1^n$ lies in two bins is $2^{-nH(\hat{X}_1)}$. This event has a vanishing probability so long as $I(X;\hat{X}_1,\hat{X}_2)<H(\hat{X}_1)$. Thus the achieved rate region is $R_1\ge I(X;\hat{X}_1,\hat{X}_2)$ such that the constraint $I(X;\hat{X}_1,\hat{X}_2)\le H(\hat{X}_1)$ and distortion constraints are satisfied.
\par 
The general coding scheme when $R_0>0$ is depicted in Fig. \ref{codebooknc} and has a \textquotedblleft doubly-binned\textquotedblright\ structure. Non-zero $R_0$ helps reduce $R_1$ by providing an extra dimension of binning. We first generate $2^{nI(X;\hat{X}_2)}$ $\hat{X}_2^n$ codewords, the indexes of which are the rows (or horizontal bins), and then in each row, we generate $2^{nI(X;\hat{X}_1|\hat{X}_2)}$ $\hat{X}_1^n$ codewords. For each row, these $\hat{X}_1^n$ codewords are then binned uniformly into $2^{nR_0}$ vertical bins, which are  the columns of our  \textquotedblleft doubly-binned\textquotedblright\ structure. Thus each bin is  \textquotedblleft doubly-indexed\textquotedblright\ (row and column index) and has a uniform number of $2^{n(I(X;\hat{X}_1|\hat{X}_2)-R_0)}$ $\hat{X}_1^n$ codewords (as in Fig. \ref{codebooknc}).  Note that this extra or independent dimension of vertical binning was not there when $R_0=0$. Intuition is that column indexing with common rate $R_0$ is independent or \textit{orthogonal} to the row indexing, and hence it  helps to reduce the private rate $R_1$.  The column or vertical bin index is described to both the decoders via common rate $R_0$ and thus $R_1$ reduces to $I(X;\hat{X}_1,\hat{X}_2)-R_0$ to describe $\hat{X}_1^n$ to Decoder 1. Here again, from knowledge of the crib, $\hat{X}_1^n$ and the column index, Decoder 2, infers the unique row index, which now will require $I(X;\hat{X}_1,\hat{X}_2)-R_0\le H(\hat{X}_1)$.
\par
We now describe the achievability in full detail.  
 \begin{figure}[htbp]
\begin{center}
\scalebox{0.70}{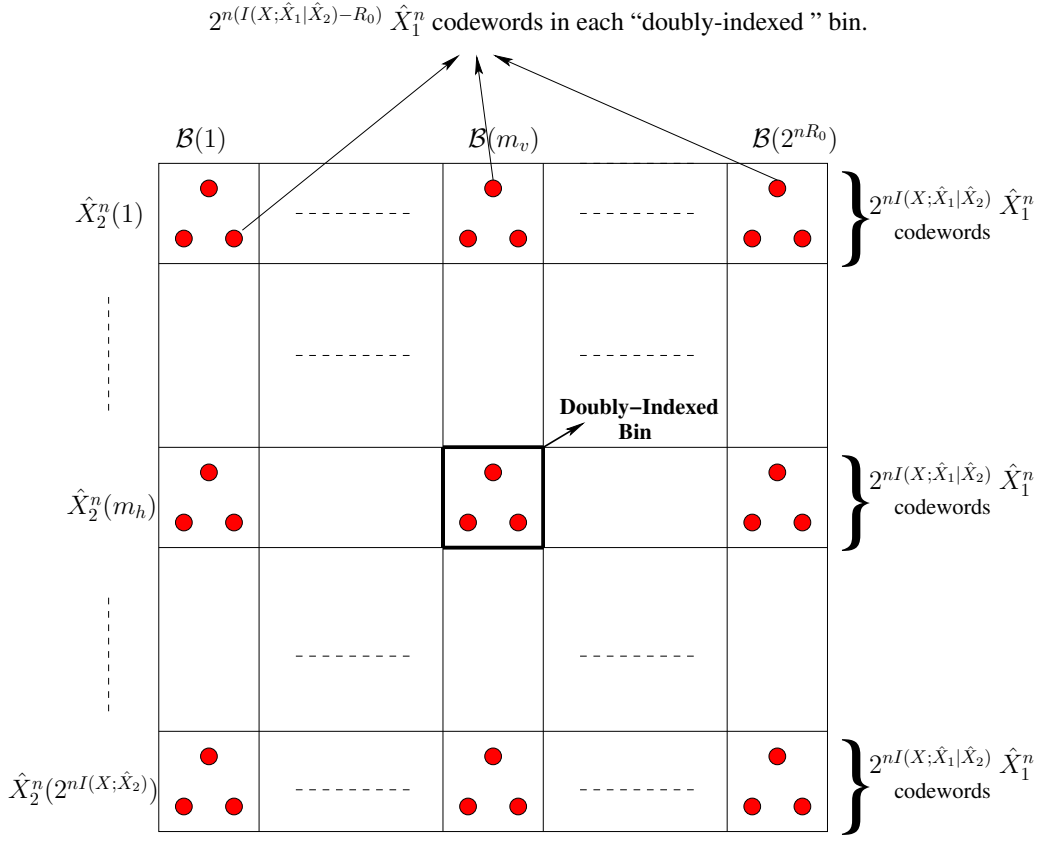}
\caption{\textquotedblleft Double Binning\textquotedblright\ - achievability scheme for the non-causal perfect cribbing.}\label{codebooknc}
\end{center}
\end{figure}

\begin{itemize}

\item \textit{Codebook Generation} : Fix the distribution $P_{X,\hat{X}_1,\hat{X}_2}$,
$\epsilon > 0$ such that
$E[d_1(X,\hat{X}_1)]\le \frac{D_1}{1+\epsilon}$ and $E[d_2(X,\hat{X}_2)]\le
\frac{D_2}{1+\epsilon}$. Generate codebook
$\mathcal{C}_{\hat{X}_2}$ consisting of $2^{nI(X;\hat{X}_2)}$ $\hat{X}^n_2(m_h)$ codewords
generated i.i.d $\sim P_{\hat{X}_2}$, $m_h\in[1:2^{nI(X;\hat{X}_2)}]$. For each
$m_h$, first generate a
codebook $\mathcal{C}_{\hat{X}_1}(m_h)$ consisting
of $2^{nI(X;\hat{X}_1|\hat{X}_2)}$ $\hat{X}^n_1$ codewords
generated i.i.d. $\sim
P_{\hat{X}_1|\hat{X}_2}$, then bin them all uniformly in $2^{nR_0}$ vertical bins $\mathcal{B}(m_v)$, $m_v\in[1:2^{nR_0}]$ and in each bin index them accordingly with $l\in[1:2^{n(I(X;\hat{X}_1|\hat{X}_2)-R_0)}]$. As outlined earlier, $m_h$ corresponds to the row or horizontal index and $m_v$ corresponds to the column or vertical index in our \textquotedblleft doubly-binned\textquotedblright\ structure, while $l$ indexes $\hat{X}_1^n$ codewords within a \textquotedblleft doubly-indexed\textquotedblright\ bin. Thus for each row and column index pair, $(m_h,m_v)$, there are $2^{n(I(X;\hat{X}_1|\hat{X}_2)-R_0)}$ $\hat{X}^n_1$ codewords. $\hat{X}_1^n$ can therefore be indexed by the triple $(m_h,m_v,l)$. The codebooks are revealed to the encoder and both the decoders.  
\item \textit{Encoding} : Given source sequence $X^n$, first the encoder finds $m_h$ from $\mathcal{C}_{\hat{X}_2}$ such that $(X^n,\hat{X}_2^n(m_h))\in\mathcal{T}^n_\epsilon$. Then the encoder finds  pair $(m_v,l)$ such that $(X^n,\hat{X}_1^n(m_h,m_v,l),\hat{X}_2^n(m_h))\in\mathcal{T}^n_\epsilon$. Thus $\hat{X}^n_1(m_h,m_v,l)\in\mathcal{B}(m_v)$. Encoder describes column or vertical bin index $m_v$ as $R_0$ to both the decoders, and the tuple $(m_h,l)$ to the Decoder 1 as rate $R_1$. Thus 
\bea
\label{eq1}
R_1\ge I(X;\hat{X}_2)+I(X;\hat{X}_1|\hat{X}_2)-R_0=I(X;\hat{X}_1,\hat{X}_2)-R_0.
\eea 
\item \textit{Decoding} : Decoder 1 knows all the indices $(m_h,m_v,l)$, and it constructs $\hat{X}_1^n=\hat{X}_1^n(m_h,m_v,l)$. Decoder 2 receives $\hat{X}_1^n$ from the non-causal cribbing and it also knows the column index $m_v$ through rate $R_0$. It then checks inside the column or vertical bin of index $m_v$, to find the unique row or horizontal bin index $m_h$ such that 
$\hat{X}_1^n=\hat{X}_1^n(m_h,m_v,\tilde{l})$ for some $\tilde{l}\in[1:2^{n(I(X;\hat{X}_1|\hat{X}_2)-R_0)}]$. The reconstruction of the Decoder 2 is then $\hat{X}_2^n=\hat{X}_2^n(m_h)$.
\item \textit{Distortion Analysis} :  Consider the following events : 

\begin{enumerate}
 \item 
 \bea
 \mathcal{E}_{1}&=&\mbox{ No $\hat{X}_2^n$ is jointly typical to a given $X^n$}\\
 &=& \bigg{\{}(X^n,\hat{X}_2^n(m_h))\notin\mathcal{T}^n_\epsilon, \forall\ m_h\in[1:2^{nI(X;\hat{X}_2)}]\bigg{\}}. 
 \eea
 The probability of this event vanishes as there are $2^{nI(X;\hat{X}_2)}$ $\hat{X}_2^n$ codewords. (cf. \textit{Covering Lemma}, Lemma \ref{covering}). 
 \item 
 \bea
 \mathcal{E}_{2}&=&\mbox{No $\hat{X}_1^n$ is jointly typical to a typical pair $(X^n,\hat{X}_2^n)$}\\
 & =&  \bigg{\{}(X^n,\hat{X}_2^n(m_h))\in\mathcal{T}^n_\epsilon\bigg{\}}\nonumber\\
 &&\cap\bigg{\{}(X^n,\hat{X}_1^n(m_h,m_v,l),\hat{X}_2^n(m_h))\notin\mathcal{T}^n_\epsilon, \forall\ m_v\in[1:2^{nR_0}], \nonumber\\ &&\forall \ l\in[1:2^{n(I(X;\hat{X}_1|\hat{X}_2)-R_0)}]\bigg{\}}.\nonumber\\ 
 \eea
 The probability of this event vanishes as corresponding to each $m_h$ there are $2^{nI(X;\hat{X}_1|\hat{X}_2)}$ $\hat{X}_1^n$ codewords, (cf. \textit{Covering Lemma}, Lemma \ref{covering}). Without loss of generality, now suppose that encoder does the encoding, $(m_h,m_v,l)=(1,1,1)$. Decoder 2 receives $\hat{X}_1^n$ via non-causal cribbing. The next two events are with respect to Decoder 2.
\item 
\bea
\mathcal{E}_3&=&\mbox{$\hat{X}_1^n$ does not lie in bin indexed by $m_h=1$ and $m_v=1$ }\\
&=&\bigg{\{}\hat{X}^n_1\neq\hat{X}^n_1(1,1,\tilde{l}),
\forall\ \tilde{l}\in[1:2^{n(I(X;\hat{X}_1|\hat{X}_2)-R_0)}]\bigg{\}}.
\eea
But the probability of this event goes to zero, because due to our encoding procedure, $\hat{X}_1^n=
\hat{X}_1^n(1,1,1)$.
\item
\bea
\mathcal{E}_4&=&\mbox{$\hat{X}_1^n$ lies in a bin with row index, $\hat{m}_h \neq 1$ and column index $m_v=1$.}\\
&=&\bigg{\{}\hat{X}^n_1=\hat{X}^n_1(\hat{m}_h,1,\tilde{l}), \ \hat{m}_h\neq 1\mbox{ for some }
\tilde{l}\in[1:2^{n(I(X;\hat{X}_1|\hat{X}_2)-R_0)}]\bigg{\}}.\nonumber\\
\eea
Since $\hat{X}_1^n=\hat{X}_1^n(1,1,1)$, this event is equivalent to finding $\hat{X}_1^n$ lying in two different rows or horizontal bins, but with the same column or vertical bin index ($m_v=1$). The probability of a single $\hat{X}_1^n$ codeword occurring repeatedly in two horizontal bins indexed with different row index is $2^{-nH(\hat{X}_1)}$, while knowing the column index, $m_v$, total number of $\hat{X}_1^n$ codewords with a particular column index are, $2^{n(I(X;\hat{X}_1,\hat{X}_2)-R_0)}$, so the probability of event $\mathcal{E}_4$ vanishes so long as, 
\bea
\label{eq1.1}
I(\hat{X}_1;\hat{X}_1,\hat{X}_2)-R_0<H(\hat{X}_1).
\eea 
\end{enumerate}
Thus consider the event, $\mathcal{E}=\mathcal{E}_1\cup\mathcal{E}_2\cup\mathcal{E}_2\cup\mathcal{E}_4$,
using the rate constraints from Eq. (\ref{eq1}) and Eq. (\ref{eq1.1}), the probability of the event vanishes if, 
\bea 
R_0+R_1&\ge& I(X;\hat{X}_1,\hat{X}_2)\\
R_0&\ge& \{I(X;\hat{X}_1,\hat{X}_2)- H(\hat{X}_1)\}^+.
\eea
We will now bound the distortion. Assume without loss of generality that,
$d_i(\cdot,\cdot)\le D_{max}$, for $i=1,2$. For both the decoders, ($i=1,2$), 
\bea
E\left[d(X^n,\hat{X}_i^n)\right]&=& P(\mathcal{E})
E\left[d(X^n,\hat{X}_i^n)|\mathcal{E}\right]+ P(\mathcal{E}^c)
E\left[d(X^n,\hat{X}_i^n)|\mathcal{E}^c\right]\\
&\stackrel{(a)}{\le}& P(\mathcal{E})D_{max}+(1+\epsilon)\E[d(X,\hat{X}_i)]\\
&\le&P(\mathcal{E})D_{max}+D_i,
\eea
where $(a)$ is via typical average lemma (cf. Typical Average Lemma \ref{typicalaverage}). Proof is completed by letting
$n\rightarrow \infty$ when $P(\mathcal{E})\rightarrow 0$.
\end{itemize}
\par
 \textit{Converse :} Converse for this setting follows by substituting $\hat{Z}_1=\hat{X}_1$ in the converse for the deterministic function cribbing in the next subsection.
\begin{note}[Joint Typicality Decoding]
Note that here our decoding for Decoder 2 relies on finding a unique bin index in which $\hat{X}_1^n$ (obtained via cribbing) lies, and there is an error if two different bins have the same $\hat{X}_1^n$. An alternative based on joint typicality decoding can also be used to achieve the same region as follows : Decoder 2 receives $\hat{X}_1^n$ via non-causal cribbing and it also knows the column index $m_v$ through rate $R_0$. It then finds the unique row or horizontal bin index $m_h$ such that $(\hat{X}_1^n,\hat{X}_1^n(m_h,m_v,\tilde{l}),\hat{X}_2^n(m_h))\in\mathcal{T}^n_\epsilon$ for some $\tilde{l}\in[1:2^{n(I(X;\hat{X}_1|\hat{X}_2)-R_0)}]$. The reconstruction of the Decoder 2 is then $\hat{X}_2^n=\hat{X}_2^n(m_h)$. We analyze the following two events, assuming without loss of generality that encoder does the encoding $(m_h,m_v,l)=(1,1,1)$.\begin{itemize}
\item 
\bea
\mathcal{E}_{d,1}&=&\mbox{Decoder 2 finds no jointly typical $\hat{X}_1^n$ indexed by $m_h=1$ and $m_v=1$ }\\
&=&\bigg{\{}(\hat{X}^n_1,\hat{X}^n_1(1,1,\tilde{l}),
\hat{X}_2^n(1))\notin\mathcal { T }
^n_\epsilon, \forall\ \tilde{l}\in[1:2^{n(I(X;\hat{X}_1|\hat{X}_2)-R_0)}]\bigg{\}}.
\eea
But the probability of this event goes to zero, because due to our encoding procedure, with high probability,  $(X^n,\hat{X}_1^n(1,1,1),\hat{X}_2^n(1))\in\mathcal{T}^n_\epsilon$. As $\hat{X}_1^n=
\hat{X}_1^n(1,1,1)$ this implies,  $(\hat{X}^n_1,\hat{X}^n_1(1,1,1),
\hat{X}_2^n(1))\in\mathcal { T }
^n_\epsilon$.
\item
\bea
\mathcal{E}_{d,2}&=&\mbox{Decoder 2 finds a jointly typical $\hat{X}_1^n$ codeword in row with index, $\hat{m}_h \neq 1$.}\\
&=&\bigg{\{}(\hat{X}^n_1,\hat{X}^n_1(\hat{m}_h,1,\tilde{l}),
\hat{X}_2^n(\hat{m}_h))\in\mathcal { T }
^n_\epsilon, \ \hat{m}_h\neq 1\mbox{ for some }
\tilde{l}\in[1:2^{n(I(X;\hat{X}_1|\hat{X}_2)-R_0)}]\bigg{\}}.\nonumber\\
\eea
By Lemma \ref{packing} (\textit{Packing Lemma }, substitute, $\card{\mathcal{A}}=2^{n(I(X;\hat{X}_1,\hat{X}_2)-R_0)}, U=\phi, X=(\hat{X}_2,\hat{X}_1), Y=\hat{X}_1$), probability of this event goes to zero with large $n$, if
\bea
\label{eq1.1}
I(\hat{X}_1;\hat{X}_1,\hat{X}_2)-R_0\le I(\hat{X}_1;\hat{X}_1,\hat{X}_2)=H(\hat{X}_1).
\eea
\end{itemize}
Thus we obtain the same constraint with the joint typicality decoding for Decoder 2. In all the subsections to follow, for Decoder 2, joint typicality decoding can also be used as an alternative to the decoding that will be described.
\end{note}

\end{proof}
\subsubsection{Deterministic Function Cribbing}
\label{subsubsec::sr_cribbing_nc_det}
\begin{theorem}
\label{theorem2}
The rate region $\mathcal{R}(D_1,D_2)$ for the setting in Fig. \ref{sr_cribbing_nc_det} with deterministic function (non-causal) cribbing is given as the closure of the set of all the rate tuples $(R_0,R_1)$ such that,
\bea
R_0+R_1&\ge&I(X;\hat{X}_1,\hat{X}_2)\\
R_0&\ge&\{I(X;\hat{Z}_1,\hat{X}_2)-H(\hat{Z}_1)\}^+,
\eea
for some joint probability distribution $P_{X}P_{\hat{Z}_1,\hat{X_2}|X}P_{\hat{X}_1|\hat{Z}_1,\hat{X}_2,X}$ such that 
$\E[d_i(X,\hat{X}_i)]\le D_i$, for $i=1,2$.
\end{theorem}
\begin{figure}[htbp]
\begin{center}
\scalebox{0.7}{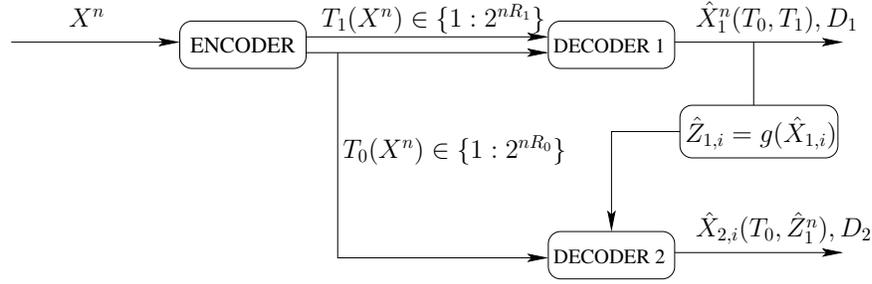}
\caption{Successive refinement, with decoders \textit{cooperating} via (deterministic function) \textit{non-causal cribbing}.}
\label{sr_cribbing_nc_det}
\end{center}
\end{figure}
\begin{proof}
\par
\textit{Achievability} : The scheme is similar to the achievability in the
previous section, where cribbing was perfect, with some minor differences. We give an outline here and highlight the differences, deferring the complete proof to Appendix \ref{appendixC}. The codebook here also has a \textquotedblleft doubly-binned\textquotedblright\ structure as in Fig. \ref{codebooknc}, the difference being that each \textquotedblleft doubly-indexed\textquotedblright\ bin has a uniform number of $\hat{Z}_1^n$ codewords instead of $\hat{X}_1^n$. So first $2^{nI(X;\hat{X}_2)}$ $\hat{X}_2^n$ codewords are generated, for each of them, $2^{nI(X;\hat{Z}_1|\hat{X}_2)}$ $\hat{Z}_1^n$ codewords are generated, which are then vertically binned uniformly into $2^{nR_0}$ vertical bins (columns).  Then for each $\hat{Z}_1^n$, $2^{nI(X;\hat{X}_1|\hat{Z}_1,\hat{X}_2)}$ $\hat{X}_1^n$ codewords are generated. Here also, the column index is described as $R_0$ and the remaining indices are  described as $R_1$, which hence is equal to $I(X;\hat{X}_1,\hat{Z}_1,\hat{X}_2)-R_0=I(X;\hat{X}_1,\hat{X}_2)-R_0$. Decoder 1 can,  as usual, construct its estimate since it knows all the indices, Decoder 2, infers the row index from the deterministic function crib, $\hat{Z}_1^n$ and knowledge of the column index. The decodability of a unique row index depends on the fact that there should not be the same $\hat{Z}_1^n$ codeword in two rows. This requires (as we saw in the previous section), $I(X;\hat{Z}_1,\hat{X}_2)-R_0\le H(\hat{Z}_1)$.
\par
\textit{Converse} : Assume we have a $(2^{nR_0},2^{nR_1},n)$ code (as per Definition \ref{definition1}) achieving respective distortions $D_1$ and $D_2$. Denote $T_1=f_{1,n}(X^n)$ and $T_0=f_{2,n}(X^n)$. Consider,
\bea
H(\hat{Z}_1^n,T_0)&\ge&I(X^n;\hat{Z}_1^n,T_0)\\
&\stackrel{(a)}{=}&I(X^n;\hat{Z}_1^n,\hat{X}_2^n,T_0)\\
&\ge&I(X^n;\hat{Z}_1^n,\hat{X}_2^n)\\
 &=&\sum_{i=1}^{n}I(X_i;\hat{Z}_1^n,\hat{X}_2^n|X^{i-1})\\
 &\stackrel{(b)}{=}&\sum_{i=1}^{n}I(X_i;\hat{Z}_1^n,\hat{X}_2^n,X^{i-1})\\
 &\ge&\sum_{i=1}^{n}I(X_i;\hat{Z}_{1,i},\hat{X}_{2,i})\\
 &\stackrel{}{=}&n\sum_{i=1}^{n}\frac{1}{n}I(X_i;\hat{Z}_{1,i},\hat{X}_{2,i})\\
 &\stackrel{(c)}{=}&nI(X_Q;\hat{Z}_{1,Q},\hat{X}_{2,Q}|Q)\\
 &\stackrel{(d)}{=}&nI(X_Q;\hat{Z}_{1,Q},\hat{X}_{2,Q},Q)\\
  &\ge &nI(X_Q;\hat{Z}_{1,Q},\hat{X}_{2,Q})\\
 H(\hat{Z}_1^n,T_0) &\le&H(\hat{Z}_{1}^n)+H(T_0)\\
 &\le&\sum_{i=1}^{n}H(\hat{Z}_{1,i})+nR_0\\
 &=&nH(\hat{Z}_{1,Q}|Q)+nR_0\\
 &\le&nH(\hat{Z}_{1,Q})+nR_0
 \eea
 \bea
  n(R_0+R_1)&=&H(T_0,T_1)\\
&=&H(T_0,T_1)-H(T_0,T_1|X^n)
\eea
 \bea
&=&I(X^n;T_0,T_1)\\
&\stackrel{(e)}{=}&I(X^n;T_0,T_1,\hat{X}_1^n,\hat{X}_2^n)\\
&=&\sum_{i=1}^{n}I(X_i;T_0,T_1,\hat{X}_1^n,\hat{X}_2^n|X^{i-1})\\
&\stackrel{(f)}{=}&\sum_{i=1}^{n}I(X_i;T_0,T_1,\hat{X}_1^n,\hat{X}_2^n,X^{i-1})\\
&\ge&\sum_{i=1}^{n}I(X_i;\hat{X}_{1,i},\hat{X}_{2,i})\\
  &=&n\sum_{i=1}^{n}\frac{1}{n}I(X_i;\hat{X}_{1,i},\hat{X}_{2,i})\\
 &\stackrel{}{=}&nI(X_Q;\hat{X}_{1,Q},\hat{X}_{2,Q}|Q)\\
 &\stackrel{}{=}&nI(X_Q;\hat{X}_{1,Q},\hat{X}_{2,Q},Q)\\ 
  &\ge &nI(X_Q;\hat{X}_{1,Q},\hat{X}_{2,Q}), 
\eea
where  (a) follows from the fact that $\hat{X}_2^n$ is a function of $(T_0,\hat{Z}_1^n)$, (b) follows from the independence of $X_i$ and $X^{i-1}$, and (c) follows by defining $Q\in [1:n]$ as a uniformly distributed time sharing random
variable independent of the source, (d) follows from the independence of $Q$ with the source process, (e) follows as $(\hat{X}_1^n,\hat{X}_2^n)$ is a function of $(T_0,T_1)$ and finally (f)  follows similarly from the independence of $X_i$ and $X^{i-1}$. Finally, we
bound the distortion as, 
\bea
D_i&\ge&\E\left[d(X^n,\hat{X}_i^n)\right]\\
&=&\E\left[\frac{1}{n}\sum_{i=1}^{n}d(X_i,\hat{X}_i)\right]\\
&=&\E[d(X_Q,\hat{X}_{i,Q})].
\eea 
The proof is completed by noting that the joint distribution of $(X_Q,\hat{X}_{1,Q},\hat{X}_{2,Q})$ 
is the same as that of $(X,\hat{X}_1,\hat{X}_2)$.
\end{proof}
\begin{note}
\label{note1}
Due to the structure of our problem, i.e., $\hat{Z}_1=g(
\hat{X}_1)$, it is easy to prove the Markov relation, $(X,\hat{X}_2)-\hat{X}_1-\hat{Z}_1$, hence the distribution mentioned in the statement of the theorem, can equivalently
be factorized as, $P_XP_{\hat{X}_1,\hat{X}_2|X}\1_{\{\hat{Z}_1=g(\hat{X}_1)\}}$, (which is the form stated in Table \ref{table_sr_cribbing}). This applies similarly for theorems to follow, and we omit this explanation henceforth.
\end{note}
\subsection{Strictly-Causal Cribbing}
\label{subsubsec::sr_cribbing_sc}
\subsubsection{Perfect Cribbing}
\label{subsubsec::sr_cribbing_nc_perfect}
\begin{theorem}
\label{theorem3}
The rate region $\mathcal{R}(D_1,D_2)$ for the setting in Fig. \ref{sr_cribbing_sc} with perfect cribbing (strictly causal) is given by the closure of the set of all the rate tuples $(R_0,R_1)$ such that,
\bea
R_0+R_1&\ge&I(X;\hat{X}_1,\hat{X}_2)\\
R_0&\ge & \{I(X;\hat{X}_1,\hat{X}_2)-H(\hat{X}_1|\hat{X}_2)\}^+,
\eea
for some joint probability distribution $P_{X,\hat{X}_1,\hat{X}_2}$ such that 
$\E[d_i(X,\hat{X}_i)]\le D_i$, for $i=1,2$.
\end{theorem}
\begin{figure}[htbp]
\begin{center}
\scalebox{0.7}{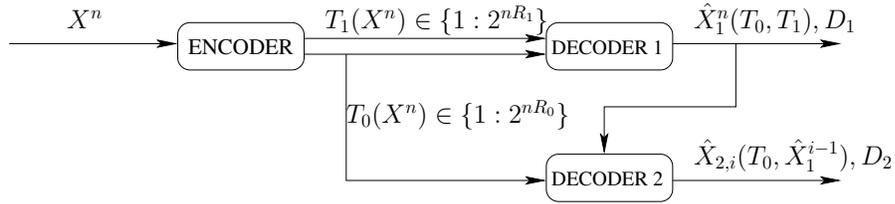}
\caption{Successive refinement, with decoders \textit{cooperating} via (perfect) \textit{strictly-causal cribbing}. }
\label{sr_cribbing_sc}
\end{center}
\end{figure}
\begin{proof}\par
\textit{Achievability} :\par
We will show the achievability of the following region instead,
\bea
R_0+R_1&\ge&I(X;\hat{X}_1,U)\\
R_0&\ge & \{I(X;\hat{X}_1,U)-H(\hat{X}_1|U)\}^+,
\eea
for some joint probability distribution $P_{X,\hat{X}_1,U}\1_{\{\hat{X}_2=f(U)\}}$. Note that the rate region in the theorem will then be obtained by simply taking $U=\hat{X}_2$. Here we deliberately present our encoding scheme with an auxiliary random variable as this will be used (with minor changes) to derive the achievable region for the case of causal cribbing discussed in the next subsection. 
\newpage
\underline{\textit{\textquotedblleft Forward Encoding\textquotedblright\ and
\textquotedblleft Block Markov Decoding\textquotedblright\ scheme}} :\par
We use a new scheme that we refer to as \textquotedblleft Forward Encoding\textquotedblright\ and
\textquotedblleft Block Markov Decoding\textquotedblright. We first briefly give an overview of the coding scheme and for simplicity consider the case when common rate $R_0=0$. Thus the source description is available only to Decoder 1, while Decoder 2  has access to the reconstruction symbols of Decoder 1, but only strictly-causally. Hence in principle we cannot deploy a scheme to operate in one block as was done for non-causal cribbing. We need to use a scheme to operate in multiple (large number) of blocks, and use an encoding procedure where $\hat{X}_1^n$ of the previous block carries information about the source sequence of the current block. In this way due to strictly causal cribbing, in the current block, Decoder 2 will know all the reconstruction symbols of Decoder 1 from the previous block, which will contain information about the source for the current block.  This is the main idea and is operated as follows : in each block, first we generate $2^{nI(X;U)}$ $U^n$ codewords, and for each $U^n$ codeword, we generate $2^{nI(X;U)}$ bins and in each bin $2^{nI(X;\hat{X}_1|U)}$ $\hat{X}_1^n$ codewords are generated. In each block, $U^n$ is jointly typical with the source sequence in the current block and the bin index describes the $U^n$ sequence jointly typical with the source sequence of the future block. This bin index carries information about the source in the future block. Hence, we address encoding as \textquotedblleft Forward Encoding\textquotedblright. Decoding is \textquotedblleft Block Markov Decoding\textquotedblright, as it assumes both decoders have currently decoded the $U^n$ sequence of the previous block. The bin index and index of the $\hat{X}_1^n$ codewords is described as $R_1$ which hence is taken to be $I(X;U)+I(X;\hat{X}_1|U)=I(X;\hat{X}_1,U)$.  Due to cribbing, Decoder 2 knows the $\hat{X}_1^n$ of the previous block and aims to find the bin index in which it lies. And as we argued in previous sections, this is possible if $I(X;\hat{X}_1,U)\le H(\hat{X}_1|U)$.
\par
 The general scheme when $R_0>0$ is depicted in Fig. \ref{codebook}. The additional step which we add to the description above (for $R_0=0$) is to bin in an extra dimension, i.e., with respect to each $U^n$ sequence we generate a \textquotedblleft doubly-binned\textquotedblright\ codebook (as in the achievability of non-causal cribbing, cf. Fig. \ref{codebooknc}). The row index encodes $U^n$ sequences of the future block and $\hat{X}_1^n$ codewords for each row are uniformly binned into $2^{nR_0}$ columns. The column index is the common description to both decoders, so $R_1$ reduces to $I(X;\hat{X}_1,U)-R_0$, and the decodability of  Decoder 2 requires the condition $I(X;\hat{X}_1,U)-R_0\le H(\hat{X}_1|U)$. 
 \par
   We now 
explain this coding scheme in detail and how it helps establish the achievable region when the cooperation between the decoders is via strictly causal cribbing.
\begin{enumerate}
 \item \textit{Codebook Generation :} The scheme does compression in blocks. Fix
the number of blocks to be $B$. In each block, $n$ source symbols are
compressed. Fix a joint probability distribution,
$P_{U,X,\hat{X}_1,\hat{X}_2}=P_{U,X,\hat{X}_1}\1_{\{\hat{X}_2=f(U)\}}$ for some
function $f$ and $\epsilon>0$ such that $\E[d_1(X,\hat{X}_1)]\le
\frac{D_1}{1+\epsilon}$ and $\E[d_2(X,\hat{X}_2)]\le
\frac{D_2}{1+\epsilon}$.  
\par 
Now in each block we generate codebook as follows. First
we generate a codebook  $\mathcal{C}_U(b)=\{u^n(b,m) \sim \prod_{i=1}^{^n}
P_{U}(u_i(b,m)),m=[1:2^{nI(X;U)}]\}$ for each block $b\in[1:B]$. For
each
$u^n(b,m)$, we create $2^{nI(X;U)}$ horizontal bins or rows $\mathcal{B}(m_h)$ which are indexed as 
$m_h\in[1:2^{nI(X;U)}]$. In each bin we
generate a codebook $2^{nI(X;\hat{X}_1|U)}$ $\hat{X}_1^n$ codewords which are then binned again into $2^{nR_0}$ vertical bins or columns,  $\mathcal{B}(m_v)$ uniformly, $m_v\in[1:2^{nR_0}]$ and index them accordingly by $l\in[1:2^{n(I(X;\hat{X}_1|U)-R_0)}]$. Thus $\hat{X}_1^n$ can be equivalently indexed as the tuple $(b,m,m_h,m_v,l)$. Hence for each $u^n$ as explained earlier we have a \textquotedblleft doubly-binned\textquotedblright\ structure, $m_h$ denotes the row index and $m_v$ denotes the column index. The codebooks are then revealed to both the encoder and decoders. 
\begin{figure}[h!]
\begin{flushleft}
\scalebox{0.5}{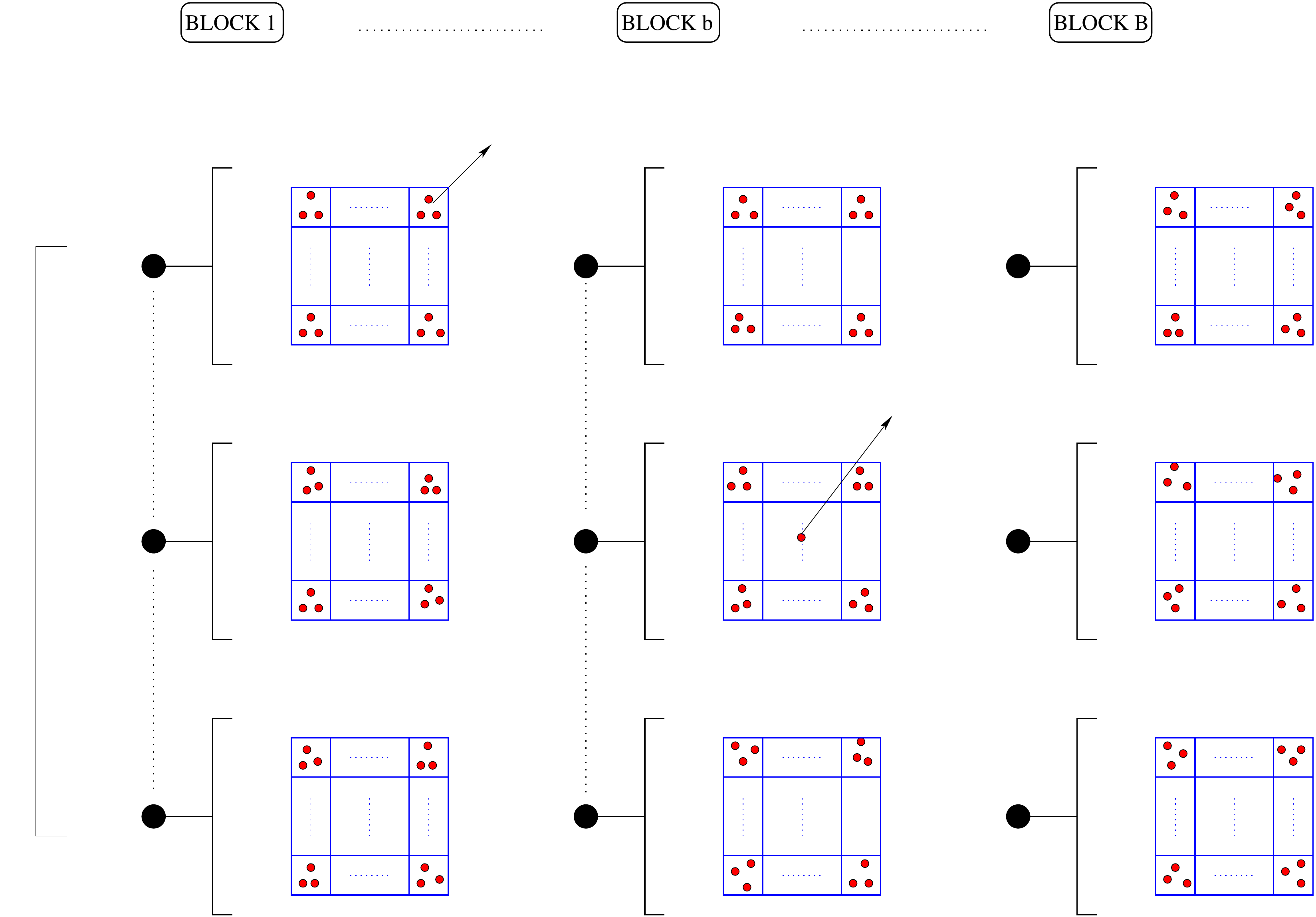}
\caption{\textquotedblleft Forward Encoding\textquotedblright\ and
\textquotedblleft Block Markov Decoding\textquotedblright\ - achievability scheme for the strictly-causal perfect cribbing.}
\label{codebook}
\end{flushleft}
\end{figure}
 \item \textit{Encoding :} $X^{nB}$ is known to the encoder. From now on additional subscripts will stand for block index, eg. $m_{h,2}$ means the row index in block 2, or $m_{v,2}$ means the column index in block 2.  Also additional scripts in parenthesis would denote the sequence in a block, eg. $X^n(b)$ will stand for the source sequence in block $b$, $\hat{X}_1^n(b)$ stands for reconstruction of Decoder 1 in block $b$. Encoding is
as follows :
\begin{enumerate}
 \item For the first block, $b=1$, assume $m_1=1$. Encoder then finds index
 $m_2$, such that $(X^n(2), U^n(2,m_2))\in\mathcal{T}^n_\epsilon$. The
encoder then looks in the codebook $\mathcal{C}_U(1)$ to find $U^n(1,m_1)$. Then
it looks in the row or horizontal bin indexed by $m_{h,1}=m_2$ corresponding to the found
$U^n(1,m_1)$, and finds the index tuple $(m_{v,1},l_1)$ such that
$(\hat{X}^n_1(1,m_1,m_{h,1},m_{v,1},l_1),X^n(1),
U^n(1,m_{1}))\in\mathcal{T}^n_\epsilon$. As found $\hat{X}_1^n\in\mathcal{B}(m_{v,1})$, the index tuple
$(m_{h,1},l_1)$ is described as $R_1$ and $m_{v,1}$ is described as $R_0$. 
 \item In the block $b$ ($\in[2:B-1]$) encoder knows $m_b$ from encoding
procedure in previous block such that $(X^n(b),U^n(b,m_b))\in
\mathcal{T}^n_\epsilon$. It then finds index $m_{b+1}$ such that
$(X^n(b+1),U^n(b+1,m_{b+1}))\in \mathcal{T}^n_\epsilon$. Now the encoder
identifies the codeword, $U^n(b,m_b)$, from the codebook $\mathcal{C}_U(b)$,
looks in the corresponding row or horizontal bin indexed as $m_{h,b}=m_{b+1}$ and finds the index tuple $(m_{v,b},l_b)$ such that
$(\hat{X}^n_1(b,m_b,m_{h,b},m_{v,b},l_b),X^n(b),
U^n(b,m_{b}))\in\mathcal{T}^n_\epsilon$. As found $\hat{X}_1^n\in\mathcal{B}(m_{v,b})$, the index tuple
$(m_{h,b},l_b)$ is described as $R_1$ and $m_{v,b}$ is described as $R_0$.  
\item In the last block $b=B$, the encoder knows $m_B$ from encoding procedure
in the previous block. Fix $m_{B+1}=1$. Encoder identifies $U^n(B,m_B)$ from
the codebook $\mathcal{C}_U(B)$, looks in the corresponding row or horizontal bin $m_{h,B}=m_{B+1}$ and finds the index tuple $(m_{v,B},l_B)$ such that
$(\hat{X}^n_1(B,m_B,m_{h,B},m_{v,B},l_B),X^n(B),
U^n(B,m_{B}))\in\mathcal{T}^n_\epsilon$. As found $\hat{X}_1^n\in\mathcal{B}(m_{v,B})$, the index tuple
$(m_{h,B},l_B)$ is described as $R_1$ and $m_{v,B}$ is described as $R_0$.  
\end{enumerate}
Hence the encoding has a \textquotedblleft Forward Encoding\textquotedblright\
interpretation, as we encoded the source sequence of the future
block as the row or horizontal bin index of the \textquotedblleft doubly-binned\textquotedblright\ codebook in the present block. As at each block $b$, $R_1$ encodes for $(m_{h,b},l_b)$, thus 
\bea
\label{eq3}
I(X;U)+I(X;\hat{X}_1|U)-R_0\le R_1.
\eea
\item \textit{Decoding :} Decoding for both the decoders is as follows :\\
\textit{Decoder 1}
\begin{enumerate}
 \item For the first block $b=1$, Decoder 1 knows $m_1=1$, and since it knows
the index $(m_{h,1},m_{v,1},l_1)$ it identifies $\hat{X}^n_1(1)=\hat{X}^n_1(1,m_1,m_{h,1},m_{v,1},l_1)$
as its source estimate for the first block. 
\item For the block $b$ ($\in[2:B]$), Decoder 1 knows $m_{b}$ from the index
sent by the encoder in the $(b-1)$ block (as $m_{h,b-1}=m_b$) and since it knows the
index $(m_{h,b},m_{v,b},l_b)$ for the current block, it identifies $\hat{X}^n_1(b)=\hat{X}^n_1(b,m_b,m_{h,b},m_{v,b},l_b)$, as its source estimate. 
\end{enumerate}
\textit{Decoder 2}
\begin{enumerate}
 \item For the first block $b=1$, Decoder 2 assumes $\hat{m}_1=1$ and generates
its estimate $\hat{X}^n_2(1)=f({U}^n(1,\hat{m}_1))$. 
\item For the block $b$ ($\in[2:B]$), Decoder 2 has already estimated
$\hat{m}_{b-1}$ in $b-1$ block. It also knows
$\hat{X}^n_1(b-1)$(because of strictly causal cribbing) and $m_{v,b-1}$ through $R_0$. It then looks into the vertical bin with index $m_{v,b-1}$ in the codebook  corresponding to the codeword $U^n({b-1},\hat{m}_{b-1})$, and finds a unique row or horizontal bin index $\hat{m}_{h,b-1}$ such that $\hat{X}^n_1(b-1)=\hat{X}^n_1(b-1,\hat{m}_{b-1},\hat{m}_{h,b-1},{m}_{v,b-1},\tilde{l}_{
b-1})$ for some $\tilde{l}_{b-1}\in[1:2^{n(I(X;\hat{X}_1|U)-R_0)}]$. But note that estimating $\hat{m}_{h,b-1}$ is equivalent to estimating $\hat{m}_{b}$, because of our forward encoding procedure, thus Decoder 2 constructs its source estimate for the block $b$ as $\hat{X}^n_2(b)=f({U}^n(b,\hat{m}_{b}))$.
\end{enumerate}
Decoding has a \textquotedblleft Block Markov Decoding\textquotedblright\
interpretation as we see that the decoding for both decoders relies on what was
successfully
decoded in the previous block. 
\item \textit{Rate Region and Bounding Distortion :} We assume without loss of
generality, $d_i(\cdot,\cdot)\le D_{max}<\infty$, $i=1,2$. In the encoding and
decoding scheme, $m_1$ was chosen to be a fixed value, deterministically chosen
prior to the compression, agreed upon by both encoders and decoders. Hence, for
both the decoders distortion in general will not be met for the first block,
however we are generous enough to allow for maximum distortion for the 
first block, which will eventually have insignificant impact on total
distortion as the number of blocks becomes large. Consider the following
encoding and decoding events which will help to
bound the distortion at Decoder 1 and Decoder 2. Suppose in block $b-1$ and $b$,
index tuples $(m_{h,b-1},l_{b-1})$ and $(m_{h,b},l_b)$ are described by the encoder
to the Decoder 1, and that $m_{h,b-1}=m_{b}$ and $m_{h,b}=m_{b+1}$, $\forall\  b=[2:B]$.
\begin{enumerate}
 \item \textit{Encoding Events} : 
\begin{itemize} \item
\bea
\mathcal{E}_{e,1}(b)&=&\mbox{No $U^n$ sequence is jointly typical with source in block $b$}\\
&=&\bigg{\{}(X^n(b),U^n(b,\tilde{m}_b))\notin\mathcal{T}
^n_\epsilon\mbox{ }\forall\mbox{ }\tilde{m}_b\in[1:2^{nI(X;U)}]\bigg{\}},
\eea 
for $b=[2:B]$. By Covering Lemma \ref{covering}, the probability of this event goes to zero as there are $2^{nI(X;U)}$ $U^n$ codewords Similarly, $P(\mathcal{E}_{e,1}(b+1))\rightarrow 0$. Suppose, $(X^n(b),U^n(b,{m}_b))\in\mathcal{T}
^n_\epsilon$ and $(X^n(b+1),U^n(b+1,{m}_{b+1}))\in\mathcal{T}
^n_\epsilon$, thus row index in block $b$ is $m_{h,b}=m_{b+1}$.
\item 
 \bea
 \mathcal{E}_{e,2}(b)&=&\mbox{No $\hat{X}_1^n$ sequence is jointly typical with the typical pair $(X,U)$ in block $b$}\nonumber\\
 &=&\mathcal{E}^c_{e,1}(b)\cap\mathcal{E}^c_{e,1}(b+1)\nonumber\\
&&\cap\bigg{\{}(\hat{X}^n_1(b,{m}_b,{m}_{h,b},\tilde{m}_{v,b},\tilde{l}_b),X^n(b),U^n(b,{m}
_b))\notin \mathcal{T}^n_\epsilon\mbox{ }\forall\mbox{
tuples }(\tilde{m}_{v,b},\tilde{t}_b)\bigg{\}},\nonumber\\
\eea
 for $b=[1:B]$, where,
 \bea
 \mathcal{E}^c_{e,1}(b)&=&\bigg{\{}(X^n(b),U^n(b,{m}_b))\in
\mathcal{T}^n_\epsilon\bigg{\}}\\ 
\mathcal{E}^c_{e,1}(b+1)&=&\bigg{\{}(X^n(b+1),U^n(b+1,{m}_{b+1}))\in
\mathcal{T}^n_\epsilon\bigg{\}}.
\eea
By Covering Lemma \ref{covering}, this event has vanishing probability as for every row index there are, $2^{nI(X;\hat{X}_1|U)}$ $\hat{X}_1^n$ codewords.
\end{itemize}

\item \textit{Decoding Events} :
Decoder 1 can perfectly construct the $\hat{X}_1^n(b)$ sequences in the block $b$. Decoder 2 in block $b$ knows $\hat{X}^n_1(b-1)$. For the Decoder 2,  for $b\in[2:B]$, assume it has decoded correctly the message, $\hat{m}_{b-1}={m}_{b-1}$ in the $b-1$ block and the encoder sends the row index $m_{h,b-1}=m_{b}$ in the block $b-1$ to Decoder 1. Also Decoder 2 knows $m_{v,b-1}$ through $R_0$. Decoder 2 needs to find an estimate $\hat{m}_{h,b-1}$, or equivalently an estimate of $\hat{m}_b$ (as $m_{h,b-1}=m_{b}$). Consider  the following events :
\begin{itemize}
\item 
\bea
\mathcal{E}_{d,1}&=&\mbox{$\hat{X}_1^n(b-1)$ does not lie in row with index, ${m}_{h,b-1}=m_b $ and column index $m_{v,b-1}$}\nonumber\\
&=&\bigg{\{}\hat{X}^n_1(b-1)=\hat{X}^n_1(b-1,{m}_{b-1} , 
{ m }
_{h,b-1},m_{v,b-1},\tilde{l}_{b-1})\bigg{\}}, 
\eea 
$\forall\ \tilde{l}\in[1:2^{n(I(X;\hat{X}_1|\hat{X}_2)-R_0)}]$. But the probability of this event goes to zero, because due to our encoding procedure, $\hat{X}_1^n(b-1)=\hat{X}^n_1(b-1,{m}_{b-1} , 
{ m }
_{h,b-1},m_{v,b-1},{l}_{b-1})
$.
\item
\bea
\mathcal{E}_{d,2}&=&\mbox{$\hat{X}_1^n(b-1)$ lies  in a row with index, $\hat{m}_{h,b-1}\neq m_b $ and column index $m_{v,b-1}$}\nonumber \\
&=&\bigg{\{}\hat{X}^n_1(b-1)=\hat{X}^n_1(b-1,{m}_{b-1} , 
\hat{ m }
_{h,b-1},m_{v,b-1},\tilde{l}_{b-1}), \ \hat{m}_{h,b-1}\neq m_b\bigg{\}}, 
\eea
$\mbox{for some } 
\tilde{l}\in[1:2^{n(I(X;\hat{X}_1|\hat{X}_2)-R_0)}]$. This event is equivalent to finding $\hat{X}_1^n(b-1)$ corresponding to $U^n(b-1,m_{b-1})$ lying in two different rows or horizontal bins, but with the same column or vertical bin index ($m_{v,b-1}$). The probability of a single $\hat{X}_1^n$ codeword (corresponding to a $U^n$ codeword) occurring repeatedly in two horizontal bins indexed with different row index is $2^{-nH(\hat{X}_1|U)}$, while knowing the column index,  total number of $\hat{X}_1^n$ codewords with a particular column index are, $2^{n(I(X;\hat{X}_1,\hat{X}_2)-R_0)}$, so the probability of event $\mathcal{E}_{d,2}$ vanishes so long as, 
\bea
\label{eqd.2}
I(\hat{X}_1;\hat{X}_1,\hat{X}_2)-R_0< H(\hat{X}_1|U).
\eea 
\end{itemize}
\end{enumerate}
Thus consider the event $\mathcal{E}(b)=\mathcal{E}_{e,1}(b)\cup\mathcal{E}_{e,2}(b)\cup\mathcal{E}_{d,1}(b)\cup\mathcal{E}_{d,2}(b)$. We have, 
\bea
P(\mathcal{E}(b))\le P(\mathcal{E}_{e,1}(b))+P(\mathcal{E}_{e,2}(b))+P(\mathcal{E}_{d,1}(b))+P(\mathcal{E}_{d,2}(b)),
\eea
 which vanishes to zero with large $n$, for each block $b=[2:B]$, if [from Eq. (\ref{eq3}), Eq. (\ref{eqd.2})], if, 
\bea
R_0+R_1&\ge& I(X;\hat{X}_1,U)\\
R_0&\ge&\{ I(X;\hat{X}_1,U)- H(\hat{X}_1|U)\}^+.
\eea
We will now bound the distortion.  The distortion for both
the decoders in the first block is bounded above by $D_{max}$. Consider the
block $b=[2:B]$ for Decoder 1, 
\bea
\E\left[d_1(X^n(b),\hat{X}_1^n(b))\right]&=&P(\mathcal{E}(b))E\left[d_1(X^n(b),\hat{X}_1^n(b))|\mathcal{E}(b)\right]\nonumber\\&&+\ P(\mathcal{E}^c(b))\E\left[d_1(X^n(b),\hat{X}_1^n(b))|\mathcal{E}^c(b)\right]\\
&\stackrel{(a)}{\le}&P(\mathcal{E}(b))D_{max}+P(\mathcal{E}^c(b))(1+\epsilon)\E[d_1(X,\hat{X}_1)]\\
&\le&P(\mathcal{E}(b))D_{max}+P(\mathcal{E}^c(b))D_1,
\eea
where (a) follows from Typical Average Lemma \ref{typicalaverage}, as given $\mathcal{E}^c(b)$, 
 $(X^n(b),\hat{X}^n_1(b))\in \mathcal{T}^n_\epsilon$. Thus as
$n\rightarrow\infty$, $P(\mathcal{E}(b))\rightarrow 0$, hence the distortion is
bounded by $D_1$ in block $b$. Similarly for
Decoder 2, 
\bea
\E\left[d_2(X^n(b),\hat{X}_2^n(b))\right]&=&P(\mathcal{E}(b))E\left[d_2(X^n(b),\hat{X}_2^n(b))|\mathcal{E}(b)\right]\nonumber\\
&&+ P(\mathcal{E}^c(b))\E\left[d_2(X^n(b),\hat{X}_2^n(b))|\mathcal{E}^c(b)\right]\\
&\stackrel{(b)}{\le}&P(\mathcal{E}(b))D_{max}+P(\mathcal{E}^c(b))(1+\epsilon)\E[d_2(X,\hat{X}_2)]\\
&\le&P(\mathcal{E}(b))D_{max}+P(\mathcal{E}^c(b))D_2,
\eea
where (b) follows from Typical Average Lemma \ref{typicalaverage}, as given $\mathcal{E}^c(b)$,
$(X^n(b),\hat{U}^n_2(b,m_b))\in \mathcal{T}^n_\epsilon$, and since
$\hat{X}_2^n(b)=f(U^n(b,m_b))$, $(X^n(b),\hat{X}^n_2(b))\in
\mathcal{T}^n_\epsilon$. 
Thus the distortion is bounded by $D_2$ in block
$b$. The total normalized distortion in $B$ blocks for Decoder 1 and Decoder 2
 is bounded above by $\frac{1}{B}D_{max}+\frac{B-1}{B}D_1$ and
$\frac{1}{B}D_{max}+\frac{B-1}{B}D_2$ respectively. Proof is completed by
letting, $B\rightarrow\infty$.
\end{enumerate}

\par
\textit{Converse} : Converse in this subsection is skipped and follows from the converse of deterministic function cribbing of the next subsection, by the substitution $\hat{Z}_1=\hat{X}_1$.
\end{proof}

\subsubsection{Deterministic Function Cribbing}
\label{subsecsec::sr_cribbing_sc_det}
\begin{theorem}
\label{theorem4}
The rate region $\mathcal{R}(D_1,D_2)$ for the setting in Fig. \ref{sr_cribbing_sc_det} with deterministic function cribbing (strictly causal) is given as the closure of the set of all the rate tuples $(R_0,R_1)$ such that,
\bea
R_0+R_1&\ge&I(X;\hat{X}_1,\hat{X}_2)\\
R_0&\ge &\{I(X;\hat{Z}_1,\hat{X}_2)-H(\hat{Z}_1|\hat{X}_2)\}^+,
\eea
for some joint probability distribution $P_{X}P_{\hat{Z}_1,\hat{X}_2|X}P_{\hat{X}_1|\hat{Z}_1,\hat{X}_2,X}$ such that 
$\E[d_i(X,\hat{X}_i)]\le D_i$, for $i=1,2$.
\end{theorem}
\begin{figure}[htbp]
\begin{center}
\scalebox{0.7}{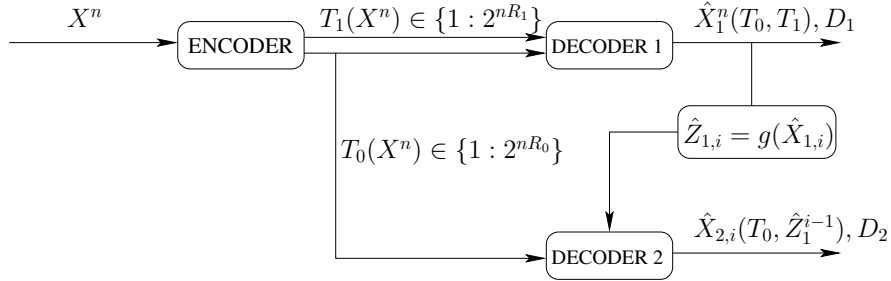}
\caption{Successive refinement, with decoders \textit{cooperating} via (deterministic function) \textit{strictly-causal cribbing}.}
\label{sr_cribbing_sc_det}
\end{center}
\end{figure}
\begin{proof}\par
\textit{Achievability} :\\
The extension to deterministic function cribbing from perfect cribbing follows similarly to the case of noncausal cribbing in Section \ref{subsubsec::sr_cribbing_nc_det}. We omit the details of achievability and describe the key idea. Here also, achievability is first proved with auxiliary random variable $U$ and the following region will be achieved,
\bea
R_0+R_1&\ge&I(X;\hat{X}_1,U)\\
R_0&\ge &\{I(X;\hat{Z}_1,U)-H(\hat{Z}_1|U)\}^+,
\eea
for some joint probability distribution $P_{X}P_{\hat{Z}_1,U|X}P_{\hat{X}_1|\hat{Z}_1,U,X}\1_{\{\hat{X}_2=f(U)\}}$ such that 
$\E[d_i(X,\hat{X}_i)]\le D_i$, for $i=1,2$. The codebook structure remains almost the same, just that instead of (uniformly) binning $\hat{X}_1^n$ into vertical $2^{nR_0}$ bins, as done in the setting of the previous subsection with perfect cribbing, we bin $\hat{Z}_1^n$ codewords and $\hat{X}_1^n$ codewords are then generated on the top of each $\hat{Z}_1^n$ codewords. Encoding changes accordingly and Decoder 2 tries to infer the row index from the deterministic crib which it obtains from Decoder 1.

\textit{Converse} :  Assume we have a $(2^{nR_0},2^{nR_1},n)$ distortion code (as per Definition \ref{definition1}) such that $(R_0,R_1,D_1,D_2)$ tuple is feasible (as per Definition \ref{definition2}). Denote $T_1=f_{1,n}(X^n)$ and $T_0=f_{2,n}(X^n)$. Identify the auxiliary random variable $U_i=(T_0,\hat{Z}_1^{i-1})$ : 
\bea
H(\hat{Z}_1^n,T_0)&\ge&I(X^n;\hat{Z}_1^n,T_0)\\
 &=&\sum_{i=1}^{n}I(X_i;\hat{Z}_1^n,T_0|X^{i-1})\\
 &\stackrel{(a)}{=}&\sum_{i=1}^{n}I(X_i;\hat{Z}_1^n,T_0,X^{i-1})\\
 &\ge&\sum_{i=1}^{n}I(X_i;\hat{Z}_{1}^i,T_0)\\
  &=&\sum_{i=1}^{n}I(X_i;\hat{Z}_{1,i},U_i)\label{ineq1}\\
 &\ge&nI(X_Q;\hat{Z}_{1,Q},U_{Q})\\
  H(\hat{Z}_1^n,T_0)&=&\sum_{i=1}^{n}H(\hat{Z}_{i,1}|T_0,\hat{Z}_1^{i-1})+H(T_0)\\
 &\le&\sum_{i=1}^{n}H(\hat{Z}_{i,1}|U_i)+nR_0\label{ineq2}\\
 &\le&nH(\hat{Z}_{1,Q}|U_Q)+nR_0,
\eea
where (a) follows from the independence of $X_i$ with $X^{i-1}$ and $Q\in [1:n]$ is similarly defined an independent (of source) uniformly distributed time sharing random
variable.  As argued in previous subsection of perfect cribbing, we lower bound $n(R_0+R_1)$ with $nI(X;\hat{X}_{1,Q},U_Q)$. Note that as $\hat{X}_{2,Q}=f(U_Q)$, for some function $f$. Lastly we bound the
distortion for both decoders as we did in previous section and note that the joint distribution of $(X_Q,\hat{X}_{1,Q},\hat{X}_{2,Q})$ is
the same as $(X,\hat{X}_1,\hat{X}_2)$ to derive the rate region with auxiliary random variable. It is easy to see that in inequalities (\ref{ineq1}) and (\ref{ineq2}), we can replace $U_i$ with $\hat{X}_{2,i}$ and this helps to provide converse for the region without auxiliary random variable provided in the theorem. 
\end{proof}

\subsection{Causal Cribbing}
\label{subsec::sr_cribbing_c_det}
\subsubsection{Perfect Cribbing}
\begin{theorem}
\label{theorem5}
The rate region $\mathcal{R}(D_1,D_2)$ for the setting in Fig. \ref{sr_cribbing_sc} with perfect causal cribbing that is $\hat{X}_{2,i}$ is a function of $(T_0,\hat{X}_1^i$), is given as the closure of the set of all the rate tuples $(R_0,R_1)$ such that,
\bea
R_0+R_1&\ge&I(X;\hat{X}_1,U)\\
R_0&\ge&\{I(X;\hat{X}_1,U)-H(\hat{X}_1|U)\}^+,
\eea
for some joint probability distribution $P_{X,\hat{X}_1,U}\1_{\{\hat{X}_2=f(U,\hat{X}_1)\}}$ such that 
$\E[d_i(X,\hat{X}_i)]\le D_i$, for $i=1,2$ and $\card{\mathcal{U}}\le\card{\mathcal{X}}\card{\mathcal{X}_1}+4$.
\end{theorem}
\begin{proof}
The achievability remains the same as in strictly causal cribbing, in terms of
encoding and decoding operations at Decoder 1. For Decoder 2, the only
change is in constructing $\hat{X}^n_2(b)$ for each block, which in this case is constructed as,
$\hat{X}^n_{2,i}(b)=f(U_i(b,m_b),\hat{X}_{1,i})$. The steps in the converse are
exactly the same as in the strictly causal cribbing case, except that this time we identify 
$\hat{X}_{2,Q}=f(U_Q,\hat{X}_{1,Q})$. The cardinality bounds on $\mathcal{U}$ follow standard arguments as in \cite{GamalKim} : $\mathcal{U}$ should have $\card{\mathcal{X}}\card{\mathcal{X}_1}-1$ elements to preserve the joint probability distribution $P_{X,\hat{X}_1}$, one element to preserve the markov chain, $(X,\hat{X}_1)-U-\hat{X}_2$, two elements to preserve the mutual information quantities, $I(X;\hat{X}_1,U)$ and $\{I(X;\hat{X}_1,U)-H(\hat{X}_1|U)\}^+$ and finally two more elements to preserve the distortion constraints.
\end{proof}

\subsubsection{Deterministic Function Cribbing}
\begin{theorem}
\label{theorem6}
The rate region $\mathcal{R}(D_1,D_2)$ for the setting in Fig. \ref{sr_cribbing_sc_det} with deterministic function cribbing but with causal cribbing, that is, $\hat{X}_{2,i}$ is a function of $(T_0,\hat{X}_1^i$), is given as the closure of the set of all the rate tuples $(R_0,R_1)$ such that,
\bea
R_0+R_1&\ge&I(X;\hat{X}_1,U)\\
R_0&\ge&\{I(X;\hat{Z}_1,U)-H(\hat{Z}_1|U)\}^+,
\eea
for some joint probability distribution $P_{X}P_{\hat{Z}_1,U|X}P_{\hat{X}_1|\hat{Z}_1,U,X}\1_{\{\hat{X}_2=f(U,\hat{Z}_1)\}}$ such that 
$\E[d_i(X,\hat{X}_i)]\le D_i$, for $i=1,2$ and $\card{\mathcal{U}}\le\card{\mathcal{X}}\card{\mathcal{X}_1}+4$.
\end{theorem}
\begin{proof}
The achievability remains the same as in strictly causal deterministic function cribbing, in terms of
encoding operation and decoding operation at Decoder 1. For the Decoder 2, only
change is in constructing $\hat{X}^n_2(b)$ for each block, it is constructed as,
$\hat{X}^n_{2,i}(b)=f(U_i(b,m_b),\hat{Z}_{1,i})$. The steps in converse are
exactly the same as in strictly causal cribbing case except that we identify, 
$\hat{X}_{2,Q}=f(U_Q,\hat{Z}_{1,Q})$. 
\end{proof}

\section{Special Cases}
\label{sec::special-case}
In this section, we study some special cases of our setting and also compute certain numerical examples.  
\subsection{The Case $R_0=0$}
\begin{figure}[htbp]
\begin{center}
\scalebox{0.7}{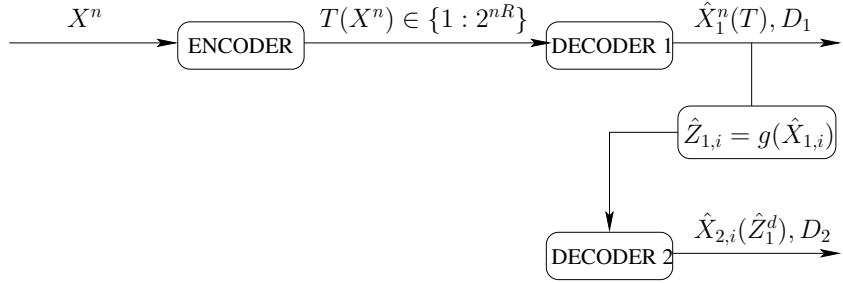}
\caption{Special case of successive refinement with cribbing decoders, when the common rate is zero. Here again $d=n$, $d=i-1$ and $d=i$ respectively stand for non-causal, strictly-causal and causal cribbing.}
\label{source_coding_cribbing}
\end{center}
\end{figure}
One special yet important case of the setting studied in previous sections, is that when $R_0=0$ as shown in Fig. \ref{source_coding_cribbing}. Here the encoder describes the source to only Decoder 1, while Decoder 2 attempts to find the reconstruction of the source within some distortion via \textit{cribbing} reconstruction symbols of the Decoder 1, non-causally, causally or strictly causally.  Table \ref{table_cribbing} provides the minimum achievable rate ($R=R(D_1,D_2)$) for various cases, derived when $R_0=0$, using Theorem 1 through 6. Distortion constraints are omitted for brevity.
\begin{table}[h!]
\begin{center}

\begin{tabular}{|l|c|c|c|c|c|}
\hline
&&&\\
$R(D_1,D_2)$ & Non-Causal ($d=n$) & Strictly-Causal ($d=i-1$) &  Causal ($d=i-1$)\\
&&&\\
\hline
&&&\\
\textit{Deterministic} & $\min I(X;\hat{X}_1,\hat{X}_2)$ & $\min I(X;\hat{X}_1,\hat{X}_2)$ & $\min I(X;\hat{X}_1,U)$\\ 
 \textit{Function} & s.t. $I(X;\hat{Z}_1,\hat{X}_2)\le H(\hat{Z}_1)$&s.t. $I(X;\hat{Z}_1,\hat{X}_2)\le H(\hat{Z}_1|\hat{X}_2)$ &s.t. $I(X;\hat{Z}_1,U)\le H(\hat{Z}_1|U)$ \\
 \textit{Cribbing}&&&\\
 & \textit{(p.m.f.)} : $P(X,\hat{X}_1,\hat{X}_2)\times$& \textit{(p.m.f.)} : $P(X,\hat{X}_1,\hat{X}_2)$& \textit{(p.m.f.)} : $P(X,\hat{X}_1,U)\times$\\
 &$ \1_{\{\hat{Z}_1=f(\hat{X}_1)\}}$&$\1_{\{\hat{Z}_1=f(\hat{X}_1)\}}$ &$\1_{\{\hat{Z}_1=f(\hat{X}_1),\hat{X}_2=f(\hat{Z}_1,U)\}}$ \\
 &&&\\
\hline
\end{tabular}
\end{center}

\caption{Results for the Successive Refinement with Cribbing Decoders, when common rate, $R_0=0$.}
\label{table_cribbing}
\end{table}
\subsection{Null $g$ function}
Our expressions reduce to the successive refinement rate region (cf. Equitz and Cover \cite{EquitzCover}), when $g$ is a trivial function. To see this consider rate region for non-causal cribbing with deterministic cribbing (cf. Theorem \ref{theorem2}), as given below,
\bea
R_0+R_1&\ge&I(X;\hat{X}_1,\hat{X}_2)\\
R_0&\ge& \{I(X;\hat{Z}_1,\hat{X}_2)-H(\hat{Z}_1)\}^+,
\eea
If $g$ is null, $\hat{Z}_1$ is constant and hence the region reduces to, 
\bea
R_0+R_1&\ge&I(X;\hat{X}_1,\hat{X}_2)\\
R_0&\ge& I(X;\hat{X}_2),
\eea
for some joint probability distribution $P_{X,\hat{X}_1,\hat{X_2}}$ such that 
$\E[d_i(X,\hat{X}_i)]\le D_i$, for $i=1,2$, which is also derived in  Equitz and Cover \cite{EquitzCover}.
\subsection{Numerical Examples}
\label{sec::examples}
We provide an example illustrating the rate regions of non-causal and strictly causal cribbing.  Along with them, the region without cribbing is also compared.  The rate regions for these three cases from the theorems in the paper are shown in the Table \ref{table_cribbing_example}. Distortion constraints are omitted for brevity.
\begin{table}[h!]
\begin{center}

\begin{tabular}{|l|c|c|c|c|c|}
\hline
&&\\
Non-Causal Cribbing & Strictly-Causal Cribbing & No Cribbing\\
&&\\
\hline
 &&\\
\small $R_0+R_1\ge I(X;\hat{X}_1,\hat{X}_2)$ & \small $R_0+R_1\ge I(X;\hat{X}_1,\hat{X}_2)$ &\small $R_0+R_1\ge I(X;\hat{X}_1,\hat{X}_2)$\\
 \small $R_0\ge \{I(X;\hat{X}_1,\hat{X}_2)- H(\hat{X}_1)\}^+$& \small $R_0\ge \{I(X;\hat{X}_1,\hat{X}_2)-H(\hat{X}_1|\hat{X}_2)\}^+$ & \small $R_0\ge I(X;\hat{X}_2)$ \\
&&\\
\textit{(p.m.f.)} : \small $P(X,\hat{X}_1,\hat{X}_2)$&\small \textit{(p.m.f.)} : $P(X,\hat{X}_1,\hat{X}_2)$& \textit{(p.m.f.)} : \small$P(X,\hat{X}_1,\hat{X}_2)$\\
&&\\
\hline
\end{tabular}
\end{center}

\caption{Comparing rate regions for the example considered, for non-causal cribbing, strictly causal cribbing and no cribbing.}
\label{table_cribbing_example}
\end{table}
\begin{figure}[htbp]
\includegraphics[scale=0.5]{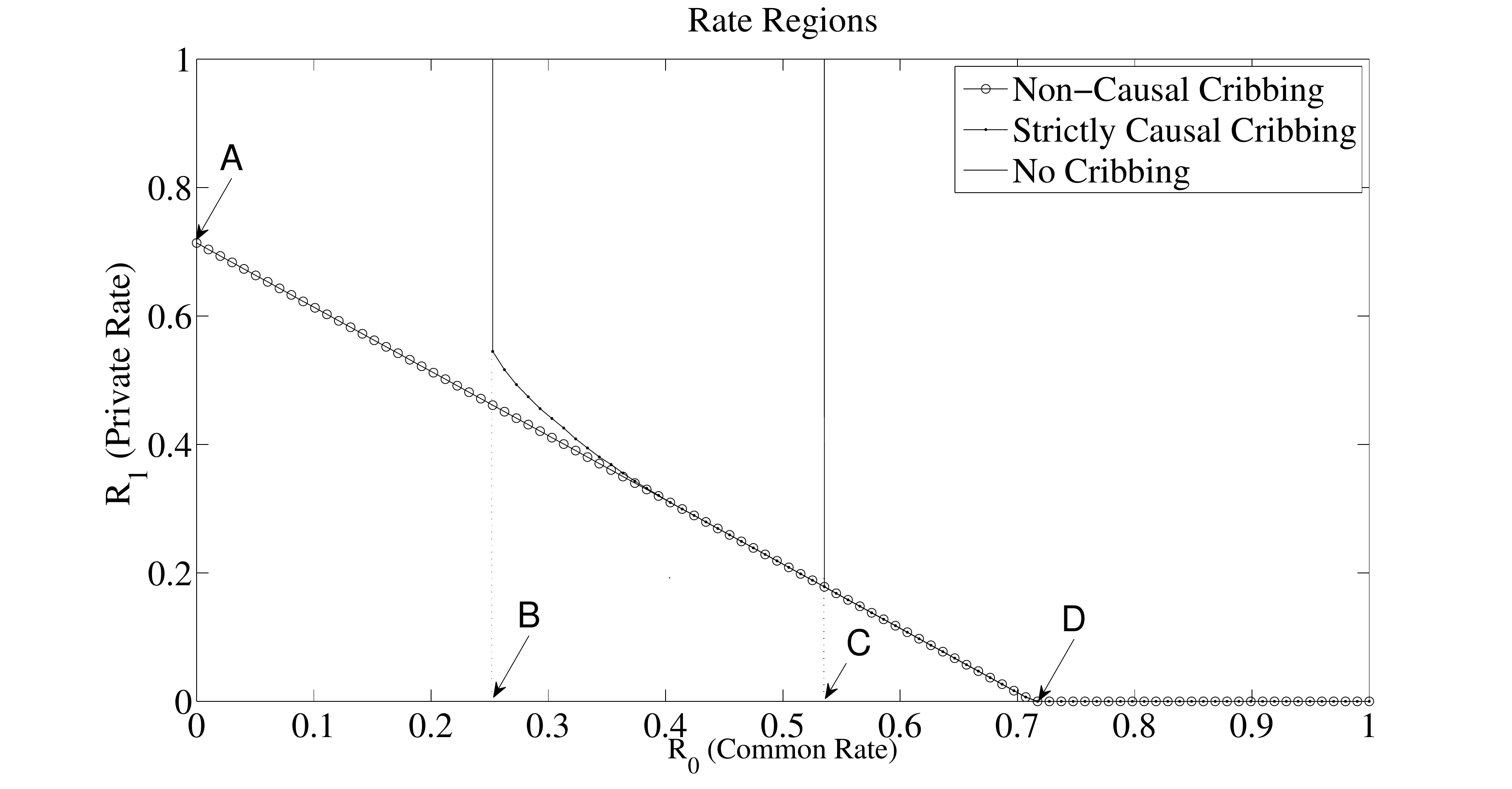}
\label{cribbingfig}
\caption{Rate regions for non-causal, strictly causal and no cribbing in successive refinement setting of Fig. \ref{sr_cribbing}. Source is $\mbox{Bern}(0.5)$ and $(D_1,D_2)=(0.05,0.1)$. The curve is tradeoff curve between $R_1$ and $R_0$ and the rate regions lie to the right of the respective tradeoff curves. }
\label{fig::examplecribbing}
\end{figure}

We plot for a specific example (cf. setting in Fig. \ref{sr_cribbing} with perfect cribbing) with a bernoulli source $X\sim \mbox{Bern}(0.5)$, binary reconstruction alphabets and hamming distortion. We consider a particular distortion tuple $(D_1,D_2)$. Due to symmetry of the source, for the optimal distribution, it is easy to argue that, $P_{\hat{X}_1,\hat{X}_2|X}(\hat{x}_1,\hat{x}_2|x)=P_{\hat{X}_1,\hat{X}_2|X}(\overline{\hat{x}_1},\overline{\hat{x}_2}|\overline{x})$, where $\overline{x}$ stands for complement of $x$.  Thus all the expressions can be written in terms of  variables $p_1=P_{\hat{X}_1,\hat{X}_2|X}(0,0|0)$, $p_2=P_{\hat{X}_1,\hat{X}_2|X}(0,1|0)$, $p_3=P_{\hat{X}_1,\hat{X}_2|X}(1,0|0)$ and  $p_4=P_{\hat{X}_1,\hat{X}_2|X}(1,1|0)$, $p_4=1-p_1-p_2-p_3$. However it is also easy to see that the distortion constraints are satisfied with equality, otherwise one can reduce the rate region slightly and still be under distortion constraint. The distortion constraints thus yield, 
\bea
\E[d(X,\hat{X}_1)]&=&p_4+p_3=D_1\\
\E[d(X,\hat{X}_2)]&=&p_2+p_4=D_2,
\eea
which implies, $p_2=1-D_1-p_1$, $p_3=1-D_2-p_1$, $p_4=p_1+D_1+D_2-1$. Thus the equivalent probability distribution space over which the closure of  rate regions is evaluated (such that distortion is satisfied) is equivalent to, $\mathcal{P}=\{p_1\in[1-D_1-D_2,\min\{1-D_1,1-D_2,2-D_1-D_2\}],p_2=1-D_1-p_1,p_3=1-D_2-p_1,p_4=p_1+D_1+D_2-1\}$. The various entropy and mutual information expressions appearing in the rate regions of non-causal, strictly causal and no cribbing (cf. Table \ref{table_cribbing_example}) can then be expressed as,
\bea
I(X;\hat{X}_1,\hat{X}_2)&=&H_2\Big{(}\Big{[}\frac{p_1+p_4}{2}\  \frac{p_2+p_3}{2}\ \frac{p_2+p_3}{2}\ \frac{p_1+p_4}{2} \Big{]}\Big{)}-H_2([p_1\ p_2\ p_3\ p_4])\\
 H(\hat{X}_1)&=&1\\
H(\hat{X}_1|\hat{X}_2)&=&H_2\Big{(}\Big{[}p_1+p_4\ p_2+p_3\Big{]}\Big{)}\\
I(X;\hat{X}_2)&=&1-H_2\Big{(}\Big{[}p_1+p_3\ p_2+p_4\Big{]}\Big{)},
\eea
where $H_2(\cdot)$ stands for the binary entropy of the probability vector. Note the only variable of optimization is effectively $p_1$. Fig. \ref{fig::examplecribbing} shows the rate regions for  $(D_1,D_2)=(0.05,0.1)$. Note that the region for no cribbing is smaller than that of strictly causal cribbing which is smaller than that of non-causal cribbing, as expected. We can also analytically compute the expression of corner points A,B,C,D in Fig. \ref{fig::examplecribbing}. Let $h_2(\alpha)=-\alpha\log\alpha-(1-\alpha)\log (1-\alpha)$ $\forall\ \alpha\in[0,1]$. Consider independent bernoulli random variables $Z_{D_1}\sim\mbox{Bern}(D_1)$ and $Z_{D_2}\sim\mbox{Bern}(D_2)$. $R_0$ for point D is evaluated by putting $R_1=0$ in rate region for non-causal cribbing and this equals $\min_{\mathcal{P}}I(X;\hat{X}_1,\hat{X}_2)$. We will now show that $\min_{\mathcal{P}}I(X;\hat{X}_1,\hat{X}_2)=1-h_2(D_1)$. Consider, $\min_{\mathcal{P}}I(X;\hat{X}_1,\hat{X}_2)\ge \min_{\mathcal{P}}I(X;\hat{X}_1)\ge \min_{\mathcal{P}}(1-H(\hat{X}_1|X))\ge 1-h_2(D_1)$, where the last two inequalities follow respectively as $\hat{X}_1$ is Bern(0.5) and that $D_1$ is the hamming distortion between $\hat{X}_1$ and $X$.  As $D_2> D_1$, this lower bound is indeed achieved if $\hat{X}_2=\hat{X}_1=X\oplus Z_{D_1}$. Similarly for point A, $R_1$ is obtained by substituting $R_0=0$ in the expression of rate region for non-causal cribbing and this again equals $ 1-h_2(D_1)$.  $R_0$ corresponding to points B and C is obtained by putting $R_1=\infty$ in the expressions of rate regions of strictly-causal and no cribbing. Let us first consider point B and observe that $R_0$ equals $\min_{\mathcal{P}}\{I(X;\hat{X}_1,\hat{X}_2)-H(\hat{X}_1|\hat{X}_2)\}^+$. We show that this equals $1-h_2(D_1)-h_2(D_2)$. To see this, consider, $\min_{\mathcal{P}}\{I(X;\hat{X}_1,\hat{X}_2)-H(\hat{X}_1|\hat{X}_2)\}^+= \min_{\mathcal{P}}\{H(\hat{X}_2)-H(\hat{X}_1,\hat{X}_2|X)\}^+\ge
\min_{\mathcal{P}}\{1-H(\hat{X}_1|X)-H(\hat{X}_1|X)\}^+$, where the last inequality follows as $\hat{X}_2$ is Bern(0.5). Since $\hat{X}_1$ and $\hat{X}_2$ are within hamming distortion $D_1$ and $D_2$ to $X$ respectively, we have $\min_{\mathcal{P}}\{I(X;\hat{X}_1,\hat{X}_2)-H(\hat{X}_1|\hat{X}_2)\}^+\ge 1-h_2(D_1)-h_2(D_2)$, where the equality holds for $\hat{X}_1=X\oplus Z_{D_1}$ and $\hat{X}_2=X+Z_{D_2}$. Similarly for point C, it can be shown $R_0$ equals $\min_{\mathcal{P}}I(X;\hat{X}_2)=1-h_2(D_2)$.

\section{Dual Channel Coding Setting }
\label{sec::duality} 
In this section we establish duality between cribbing decoders in the successive
refinement problem and cribbing encoders in the MAC problem with a common message. The
duality between rate-distortion and channel capacity was first mentioned by Shannon,
\cite{Shannon60} and was further developed for the case of side information by Pradhan et. al.,
\cite{Pradhan_duality03} and by Chiang and Cover, \cite{Cover_chiang_duality02}. Additional
duality has been shown by Yu, \cite{Yu_duality03} for a class of broadcast channels and
multiterminal source coding problems, and by Shirazi et. al., \cite{ShiraziPermuter10_Allerton} for the
case of increased partial side information. The duality between source and channel coding with action dependent side information was shown in Kittichokechai et al. in \cite{Kittichokechai}. Recently, Gupta and Verd\'{u},
\cite{GuptaVerdue_Duality10} have shown operational duality between the codes of source coding
and of channel coding with side information.

To make the notion of duality clearer and sharper, we consider coordination problems in source coding \cite{CuffPermuterCover09IT_coordination_capacity} and for channel coding
we consider a new kind of problems which we refer to as channel coding with restricted code distribution.
In the (weak) coordination problem \cite{CuffPermuterCover09IT_coordination_capacity} the goal is
to generate a joint typical distribution of the sources and the reconstruction (or actions)
rather than a distortion constraint between the source and its reconstruction. Similarly, we
define a channel coding problem where the code is restricted to a specific type. The
achievability proofs for coordination and channel capacity with restricted code distribution are 
the same as that of rate-distortion and channel capacity, respectively, since the codes in all
achievability proofs are generated randomly with  specific distribution. The converse is also
similar except in the last step where we need to justify the constraint of having a code with a
specific type. For this purpose we invoke  \cite[Property 2]{CuffPermuterCover09IT_coordination_capacity} that is stated as follows : 

\begin{lemma}[Equivalence of type and time-mixed variables  \cite{CuffPermuterCover09IT_coordination_capacity}]\label{l_XqYqZq_type}  For a collection of random sequences $X^n$, $Y^n$, and $Z^n$, the expected joint type ${\bf
E} P_{X^n,Y^n,Z^n}$ is equal to the joint distribution of the time-mixed variables
$(X_Q,Y_Q,Z_Q)$, where $Q$ is a r.v. uniformly distributed over the integers $\{1,2,3,...,n\}$  and independent of $(X^n, Y^n, Z^n)$.
\end{lemma}

The duality principle between source coding and channel coding with cribbing appears later in
Table \ref{t_duality}. According to those principles, the standard successive refinement source
coding problem which was introduced in \cite{EquitzCover} is dual to the MAC with one common
message and one private message \cite{Slepian_Wolf_MAC73}. Furthermore, the successive refinement
source coding with cribbing decoders is dual to the MAC with one common message and one
private message and cribbing encoders.
 To show the duality, let us investigate the
capacity of the MAC with common message and cribbing encoders and compare it to the
achievable region of the successive refinement problem with cribbing.

\subsection{MAC with cribbing encoders and a common message}
 We consider here the problem of MAC with partial cribbing encoders where there is one private
 message $m_1\in\{1,2,...,2^{nR_1}\}$ known to Encoder 1 and one common message $m_0\in\{1,2,...,2^{nR_0}\}$ known to both encoders that needs
 to be sent
 to the decoder, as shown in Fig. \ref{duality} . We assume that Encoder 2 cribs the signal from
 Encoder 1, namely, Encoder 2 observes a deterministic function of the output of Encoder 1. We
 consider here three cases, noncausal, strictly-causal and causal cribbing and we show in the next subsection their duality to the
 successive refinement problem.

\begin{definition}
\label{definition1}\label{def_mac_code} A ($2^{nR_0},2^{nR_1},n, P(x_1,x_2)$) partial cribbing
MAC, with one private and one common message and a code restricted to a distribution
$P(x_1,x_2)$, has, 
\begin{enumerate}
\item Encoder 1, $g_{1} : \{1, . . . , 2^{nR_0}\}\times\{1,...,2^{nR_1}\} \rightarrow {\mathcal{X}}_1^n$.
\item Encoder 2, $\forall\
i=1,...,n$. (depending on $d$ in Fig. \ref{sr_cribbing}, the decoder mapping changes as below),
\begin{eqnarray}
g^{nc}_{2,i} &:& \{1, . . . , 2^{nR_0}\}\times{\mathcal{Z}}_1^n \rightarrow {\mathcal{X}}_{2}\mbox{\textit{\ \ \ \ non-causal cribbing, $d=n$}} \\
g^{sc}_{2,i} &:& \{1, . . . , 2^{nR_0}\}\times{\mathcal{Z}}_1^{i-1} \rightarrow {\mathcal{X}}_{2}\mbox{\textit{\ \ \ \ strictly-causal cribbing, $d=i-1$}}, \\
g^{c}_{2,i} &:& \{1, . . . , 2^{nR_0}\}\times{\mathcal{Z}}_1^i \rightarrow
{\mathcal{X}}_{2}\mbox{\textit{\ \ \ \ causal cribbing, $d=i$}}.
\end{eqnarray}

\item Decoder,  $f: \mathcal{Y}^n \rightarrow \{1, . . . , 2^{nR_0}\}\times \{1, . . . , 2^{nR_1}\}$.
\end{enumerate}
\end{definition}
An error occurs if the one of the messages was incorrectly decoded or if the joint type of the
output and input to the channel deviates from the required one. Hence, the probability of error is
defined for any integer $n$ and $\delta>0$ such as
\begin{equation}
Pe^{(n),\delta} = \Pr\left\{ (\hat M_0(Y^n), \hat M_1(Y^n))\neq (M_0,M_1) \text{ AND }
\|P_{X_1^n,X_2^n,Y^n}(x,y,z)-P(x_1,x_2)P(y|x_1,x_2)\|_{TV}\geq \delta\right\},
\end{equation}
where $P_{X_1^n,X_2^n,Y^n}(x,y,z)$ is the joint type of the input and output of the channel and
$\|\cdot \|_{TV}$ is  the total variation between two probability mass functions, i.e.,  half the
$L_1$ distance between them, given by
\begin{eqnarray*}
\|p(x,y,z) - q(x,y,z)\|_{TV} & \triangleq & \frac{1}{2} \sum_{x,y,z} |p(x,y,z) - q(x,y,z)|.
\end{eqnarray*}
A pair rate $(R_0,R_1)$ is achievable if for any $\delta>0$ there exists a sequence of codes such
that $Pe^{(n),\delta}\to 0$  as $n\to \infty$. The capacity region is defined in the standard way
for MAC as in \cite[Chapter 15.3]{CovThom06}, as the union of all achievable rate pairs. Let us
define three regions $\mathcal R^{nc}, \mathcal R^{sc}$ and $\mathcal R^{c}$, which correspond
to noncausal, strictly-causal, and causal cases.
\begin{equation}\label{e_Rnc}
\mathcal R^{nc}(P)\triangleq 
\left\{
\begin{array}{l}
R_1\leq I(Y;X_1|X_2,Z_1)+H(Z_1) \\
R_0+R_1\leq I(Y;X_1,X_2),
\end{array}
\right.
\end{equation}
\begin{equation}\label{e_Rsc}
\mathcal R^{sc}(P)\triangleq
\left\{
\begin{array}{l}
R_1\leq I(Y;X_1|X_2,Z_1)+H(Z_1|X_2) \\
R_0+R_1\leq I(Y;X_1,X_2).
\end{array}
\right.
\end{equation}
\begin{equation}\label{e_Rc}
\mathcal R^{c}(P)\triangleq \bigcup_{P(u|x_1){\bf{1}}_{x_2=f(u,z_1)}} \left\{
\begin{array}{l}
R_1\leq I(Y;X_1|U,Z_1)+H(Z_1|U) \\
R_0+R_1\leq I(Y;X_1,U),
\end{array}
\right.
\end{equation}
where the union is over joint distributions that preserve the constraint $P(x_1,x_2)$. Since
$x_2=f(u,z_1)$, note that $I(Y;X_1,U)=I(Y;X_1,X_2)$.
The next theorem states that the regions defined above, $\mathcal R^{nc}(P), \mathcal R^{sc}(P)$ and
$\mathcal R^{c}(P)$  are the respective capacity regions. 

\begin{theorem}[MAC with common message and cribbing encoders] \label{t_mac} The capacity regions of
MAC with common message, restricted code distribution $P(x_1,x_2)$  and non-causal,
strictly-causal and causal cribbing that is depicted in Fig. \ref{duality} are $\mathcal
R^{nc}(P), \mathcal R^{sc}(P)$ and $\mathcal R^{c}(P),$ respectively.
\end{theorem}

The achievability and the converse proof of the theorem is presented in the Appendix. In the coding
scheme of the achievability proof, we use block Markov coding, backward decoding and rate
splitting similar to the techniques used in Willems and Van der Muelen \cite{Willems_Cribbing} and Permuter and Asnani 
\cite{HaimHimanshuCribbing}. The converse uses the standard Fano's inequalities and the
identification of an auxiliary random variable.

\subsection{Duality results between successive refinement and MAC with a common message }
\begin{table}[h!]
\begin{center}
\begin{tabular}{||l|l||}

\hline\hline
 SOURCE CODING & CHANNEL CODING \\ \hline\hline
Source encoder & Channel decoder\\ \hline
Encoder input  $X_i$ & Decoder input $Y_i$ \\\hline
Encoder output  & Decoder output \\
$M\in\{1,2,..,2^{nR}\}$ & $M\in\{1,2,..,2^{nR}\}$\\\hline
Encoder function  & Decoder function  \\
$f:\mathcal X^n\mapsto  \{1,2,...,2^{nR}\}$ & $f:\mathcal X^n\mapsto \{1,2,..,2^{nR}\}$\\\hline 
Source decoder input & Channel encoder input\\
$M\in\{1,2,..,2^{nR}\}$ & $M\in\{1,2,..,2^{nR}\}$ \\\hline
Decoder output  $\hat X^n$  & Encoder output  $X^n$\\\hline
Cribbing decoders $\hat{Z}_i(\hat X_{i})$ &Cribbing encoders  ${Z}_i(X_{i})$\\\hline
Noncausal cribbing decoder & Noncausal cribbing encoder\\
 $f_i:\{1,2,...,2^{nR}\}\times \hat Z^n\mapsto \hat X_i$ & $f_i:\{1,2,...,2^{nR}\}\times Z^n\mapsto X_i$ \\\hline
Strictly-causal cribbing decoder & Strictly-causal cribbing encoder
\\
$f_i:\{1,2,...,2^{nR}\}\times \hat Z^{i-1}\mapsto \hat X_i$ & $f_i:\{1,2,...,2^{nR}\}\times Z^{i-1}\mapsto X_i$ \\\hline
Causal cribbing decoder &Causal cribbing encoder \\
$f_i:\{1,2,...,2^{nR}\}\times \hat Z^i\mapsto \hat X_i$ & $f_i:\{1,2,...,2^{nR}\}\times Z^i\mapsto X_i$ \\\hline
Auxiliary r.v. $U$ & Auxiliary r.v. $U$\\\hline
Constraint  & Constraint  \\
$P(x,\hat x_1,\hat x_2)$, $P(x)$ is fixed & $P(y, x_1, x_2)$, $P(y|x_1,x_2)$ is fixed\\\hline
Joint distribution $P(x,\hat x_1,\hat x_2,u)$ & Joint distribution  $P(y, x_1,x_2,u)$\\\hline

 \hline\hline
\end{tabular}
\end{center}
\caption{Principles of duality between source coding and channel coding}\label{t_duality}
\end{table}
\begin{figure}[htbp]
\begin{center}
\scalebox{0.55}{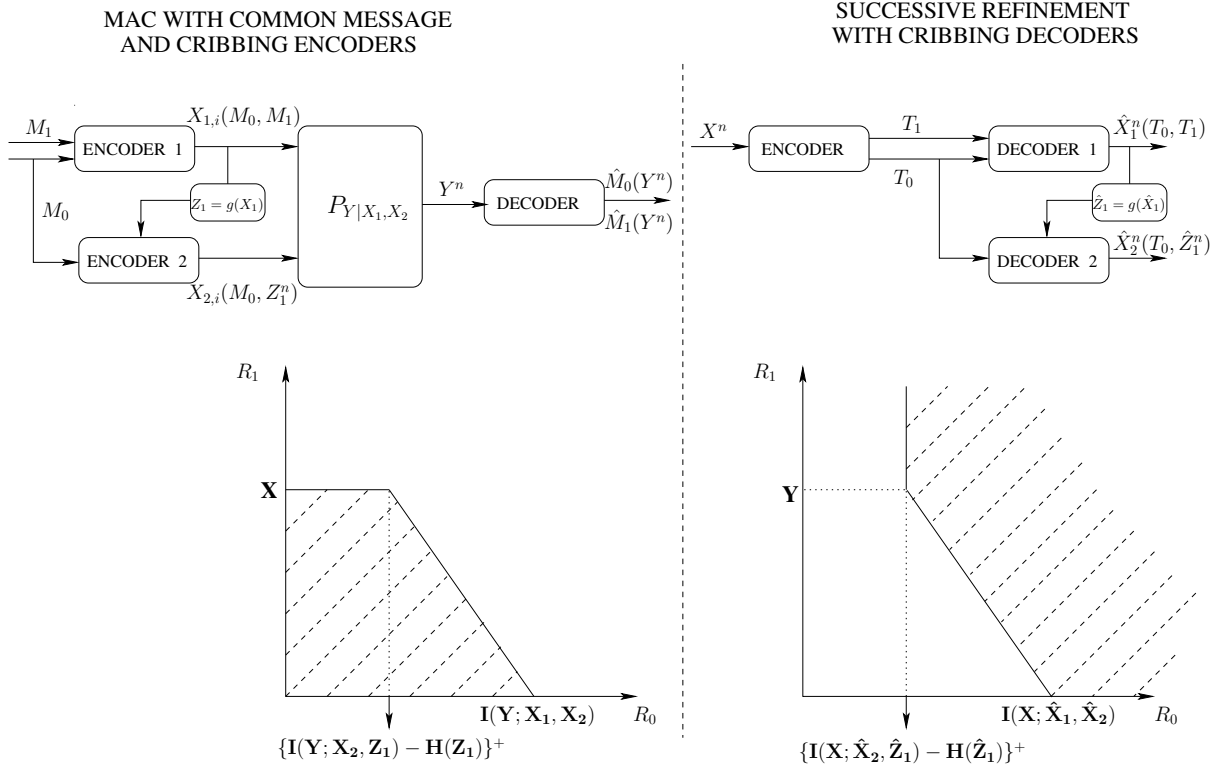}
 \caption{Duality between the cribbing decoders in
successive refinement problem and the cribbing encoders in the MAC problem with a common message,
non-causal case. Table \ref{t_duality} represents how the expression of rate and capacity regions
of the two problems are related. In the figure, for a fixed joint probability distribution, we plot the rate and capacity regions, and we observe that the
corner points are dual to each other. Point $\mathbf{Y}$ corresponds to $(R_0,R_1)=(0, I(X;\hat{X}_1,\hat{X}_2)-\{I(X;\hat X_2,Z_1)-H(Z_1)\}^+)$ and Point $\mathbf{X}$ corresponds to $(R_0,R_1)=(0, I(Y;{X}_1,{X}_2)-\{I(Y;X_2,Z_1)-H(Z_1)\}^+).$ }\label{duality}
\end{center}
\end{figure}
Now that we have the capacity regions of the MAC with common message and of successive refinement
we  explore the duality of the regions. From a first glance at the  regions of MAC with common
message and of successive refinement, their duality may go unnoticed. However, the corner points of the
regions are dual according to the principles presented in Table \ref{t_duality} and as seen in
Fig. \ref{duality}.
\begin{table}[h!]
\begin{center}
\begin{tabular}{||c|c||}
\hline\hline
& Corner points $(R_0,R_1)$ of the noncausal ($d=n$) \\
\hline\hline
 MAC  & $(I(Y;{X}_1,{X}_2),0)$ \\
 Eq. (\ref{e_Rnc}) &$(\{I(Y;X_2,Z_1)-H(Z_1)\}^+,I(Y;{X}_1,{X}_2)-\{I(Y;X_2,Z_1)-H(Z_1)\}^+) $ \\
  \hline
  SR  & $(I(X;\hat{X}_1,\hat{X}_2),0)$ \\
 Theorem \ref{theorem2} &$(\{I(X;\hat X_2,Z_1)-H(Z_1)\}^+,I(X;\hat{X}_1,\hat{X}_2)-\{I(X;\hat X_2,Z_1)-H(Z_1)\}^+) $ \\
  \hline \hline
\end{tabular}
\end{center}

\caption{The corner points of the noncausal case.} \label{table_cmp_cribbing_nc}
\end{table}

\begin{table}[h!]
\begin{center}

\begin{tabular}{||c|c||}
\hline\hline
& Corner points $(R,R_1)$ of  the strictly causal case ($d=i-1$) \\
\hline\hline
 MAC  & $(I(Y;{X}_1,X_2),0),$ \\
 Eq. (\ref{e_Rsc}) &$(\{I(Y;X_2,Z_1)-H(Z_1|X_2)\}^+,I(Y;{X}_1,U)-\{I(Y;X_2,Z_1)-H(Z_1|X_2)\}^+) $ \\
  \hline
  SR  & $(I(X;\hat{X}_1,\hat X_2,0)$ \\
 Theorem \ref{theorem4} &$(\{I(X;\hat X_2,Z_1)-H(Z_1|\hat X_2)\}^+,I(X;\hat{X}_1,\hat X_2)-\{I(X;\hat X_2,Z_1)-H(Z_1|\hat X_2)\}^+) $ \\
  \hline \hline
\end{tabular}
\end{center}

\caption{The corner ponts of the strictly causal case.} \label{table_cmp_cribbing_sc}
\end{table}
\begin{table}[h!]
\begin{center}

\begin{tabular}{||c|c||}
\hline\hline
& Corner points $(R,R_1)$ of  the causal case ($d=i$) \\
\hline\hline
 MAC  & $(I(Y;{X}_1,U),0),$ where $X_2=f(U,Z_1)$ \\
 Eq. (\ref{e_Rc}) &$(\{I(Y;U,Z_1)-H(Z_1|U)\}^+,I(Y;{X}_1,U)-\{I(Y;U,Z_1)-H(Z_1|U)\}^+) $ \\
  \hline
  SR  & $(I(X;\hat{X}_1,U,0)$ where $\hat X_2=f(U,Z_1)$ \\
 Theorem \ref{theorem6} &$(\{I(X;U,Z_1)-H(Z_1|U)\}^+,I(X;\hat{X}_1,U)-\{I(X;U,Z_1)-H(Z_1|U)\}^+) $ \\
  \hline \hline
\end{tabular}
\end{center}

\caption{The corner points of the causal case.} \label{table_cmp_cribbing_nc}
\end{table}
Tables \ref{table_cmp_cribbing_nc}-\ref{table_cmp_cribbing_sc} presents the corner points of the
capacity region of the MAC with partial cribbing and common message and compare them to the
corner points of the successive refinement (SR) rate region with partial cribbing encoders. Note that
applying the dual rules $X_1\leftrightarrow\hat X_1$, $X_2\leftrightarrow\hat X_2$,
$Y\leftrightarrow X$, and $\geq \leftrightarrow \leq$, we obtain duality between the corner
points of the capacity region of MAC with common message and the rate region of the successive
refinement setting.

\subsection{Duality between MAC with conferencing encoders and successive refinement with
conferencing decoders} In the previous subsection we saw that there is a duality between the
problem of MAC with one common message and one private message with cribbing encoders to
successive refinement with cribbing decoders. Now we  show that the duality also exists if the
cooperation between the encoders/decoders is through a limited rate (conferencing) link as shown in Fig. \ref{duality_conferencing}.
\begin{theorem}\label{t_MAC_conf}
The capacity region of MAC with one common message  at rate $R_0$ known to both encoders, one
private message at rate $R_1$ known to Encoder 1, and a limited rate link from Encoder 1 to
Encoder 2 at rate $R_{12}$ with a restricted code distribution $P(x_1,x_2)$ is
\begin{eqnarray}
R_0+R_1&\leq& I(X_1,X_2;Y)\nonumber \\
R_1&\leq& I(X_1;Y|X_2)+R_{12}.
\end{eqnarray}
\end{theorem}
This theorem can be proved using the  result of conferencing MAC \cite{Willems83_cooperating}
where $C_{21}=\infty$, and  choosing $U=X_2$. It is also
possible to prove the theorem directly. The achievability part of Theorem \ref{t_MAC_conf} follows easily if $R_1\leq R_{12}$, then the conferencing link can be used to convey message $M_1$, thus both the encoders have a common knowledge of both the messages, so $R_0+R_1\leq I(X_1,X_2;Y)$ is achievable. If rate $R_{12}\leq R_1$, then the conferencing can be used to increase the
common message rate to $R_0+R_{12}$ and decrease the private message rate to $R_1-R_{12}$. The
converse can be proved using the fact that $nR_1=H(M_1)=H(M_1|M_0)\leq
H(M_{12}|M_0)+H(M_1|M_0,M_{12})$, and then bounding $H(M_{12}|M_0)\leq nR_{12}$ using the fact that the
cardinality of $M_{12}$ is $2^{nR_{12}}$ and bounding $H(M_1|M_0,M_{12})\leq nI(X_{1Q};Y_Q|X_{2Q})$  using Fano's
inequality and the fact that the channel is memoryless.

Finally, one can note a duality between the MAC with one common message and one private message
and conferencing encoders to the successive refinement with conferencing decoders. In particular,
Table \ref{table_corner_conf} presents the corner points of the achievability regions of the two
problems from which the duality rules $X_1\leftrightarrow\hat X_1$, $X_2\leftrightarrow\hat X_2$,
$Y\leftrightarrow X$,  emerge.

\begin{figure}[htbp]
\begin{center}
\scalebox{0.55}{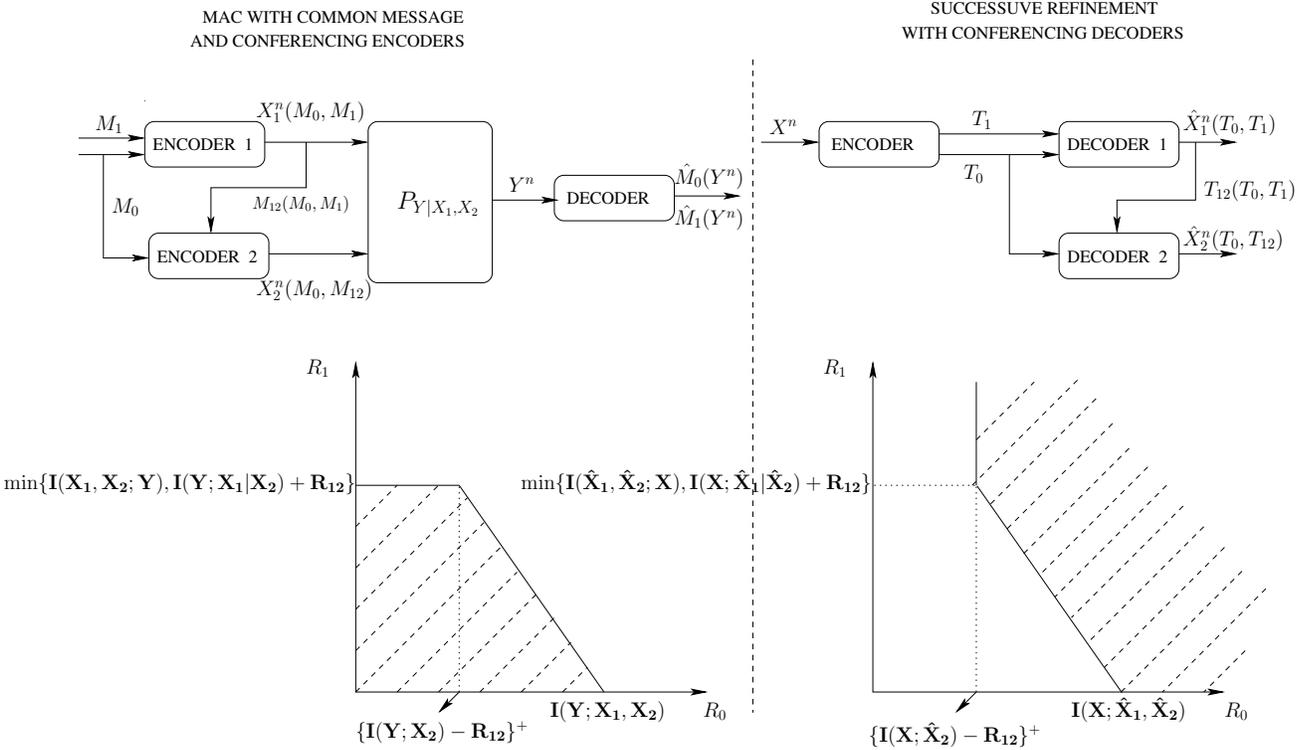}
 \caption{Duality between the conferencing decoders in
successive refinement problem and the conferencing encoders in the MAC problem with a common message,
non-causal case. In the figure, for a fixed joint probability distribution, we plot the rate and capacity regions, and we observe that the
corner points are dual to each other. } 
 \label{duality_conferencing}
\end{center}
\end{figure}

\begin{table}[h!]
\begin{center}
\begin{tabular}{||c|c||}
\hline\hline
& Corner points $(R_0,R_1)$ of  the conferencing case \\
\hline\hline
 MAC  & $(I(Y;{X}_1,X_2),0)$  \\
 Theorem \ref{t_MAC_conf} &$(\min(I(Y;{X}_1,X_2), I(Y;X_1|X_2)+R_{12}), \{I(Y;X_2)-R_{12}\}^+) $ \\
  \hline
  SR  & $(I(X;\hat{X}_1,\hat{X}_1),0)$ \\
 Eq. (\ref{eqpre1})-(\ref{eqpre2}) &$(\min(I(X;\hat {X}_1,\hat X_2), I(X;\hat X_1|\hat X_2)+R_{12}), \{I(X;\hat X_2)-R_{12}\}^+) $  \\
  \hline \hline
\end{tabular}
\caption{The corner points of the achievabilities of the MAC with one common message and one
private message and conferencing encoders and of successive refinement with conferencing
decoders.} \label{table_corner_conf}
\end{center}
\end{table}

\section{conclusion}
\label{sec::conclusion}	
In this paper, we introduced new models of cooperation in multi terminal source coding. The setting of successive refinement with single encoder and two decoders was generalized to incorporate cooperation between the users via (a) \textit{conferencing}, or (b) \textit{cribbing}. A new scheme,\textquotedblleft Forward Encoding\textquotedblright\ and
\textquotedblleft Block Markov Decoding\textquotedblright\ was used to derive the rate regions for strictly-causal and causal cribbing. Certain numerical examples are presented and show how cooperation via cribbing can boost the rate region. Finally, we introduce dual channel coding problems, and establish duality between successive refinement with cribbing decoders and communication over the MAC with common message and cribbing encoders.

\bibliography{cribbingbib}
\bibliographystyle{IEEEtran}

\appendices
\section{Successive refinement with Conferencing Decoders, Fig. \ref{sr_conferencing}}
\label{sec::sr_conferencing}
Consider Fig. \ref{sr_conferencing}, here Decoder 1 cooperates with Decoder 2 by providing an additional description $T_{12}$ to it. The rate region is given by, 
\bea
R_0+R_1&\ge& I(X;\hat{X}_1,\hat{X}_2)\\
R_0+R_{12}&\ge& I(X;\hat{X}_2),
\eea
for some joint probability distribution $P_{X,\hat{X}_1,\hat{X}_2}$ such that $\E[d_i(X_i,\hat{X}_i)]\le D_i$, for $i=1,2$. We will briefly describe the proof as they are based on standard arguments used throughout the paper. 
\par
\textit{Achievability} : We provide the achievability under two cases,
\begin{itemize}
\item \textit{Case 1} :  $R_1\le R_{12}$, here we describe $T_1$ through $T_{12}$, thus both the decoders know $(T_0,T_1)$ and hence the following region is achievable,
\bea
R_1&\le& R_{12}\\
R_0+R_1&\ge&I(X;\hat{X}_1,\hat{X}_2),
\eea
for a joint probability distribution $P_{X,\hat{X}_1,\hat{X}_2}$ such that distortion constraints are satisfied. The region is equivalent to, (call it Region 1)
\bea
R_1&\le& R_{12}\\
R_0+R_1&\ge&I(X;\hat{X}_1,\hat{X}_2)\\
R_0+R_{12}&\ge&I(X;\hat{X}_2).
\eea
\item \textit{Case 2} : $R_1> R_{12}$, here $T_1$ is described as a tuple $(T_1',T_1'')$ of rate $(R_1-R_{12},R_{12})$, and $T_1''$ is described via the conferencing link. Thus this problem is similar to original successive refinement problem, where encoder has a private rate $R_1-R_{12}$ and a common rate $R_0+R_{12}$, and hence the following region (call Region 2) is achievable (follows from the achievability of Equitz and Cover \cite{EquitzCover}),
\bea
R_1&>& R_{12}\\
R_0+R_1&\ge&I(X;\hat{X}_1,\hat{X}_2)\\
R_0+R_{12}&\ge&I(X;\hat{X}_2),
\eea
for a joint probability distribution $P_{X,\hat{X}_1,\hat{X}_2}$ such that distortion constraints are satisfied. 
\end{itemize}
We finish the proof of achievability by combining Region 1 and Region 2. Converse follows from standard cutset bound arguments and is omitted.
\section{Proof of Lemma \ref{lemma1}}
\label{appendix_cascade}
The proof is done by proving set inclusions in two directions as done for Theorem 3 in \cite{Vasudevan_Diggavi_Tian_Cascade}. First we prove, $\tilde{\mathcal{R}}(D_1,D_2)\subseteq \mathcal{R}_{cascade}(D_1,D_2)$. Suppose a pair $(\tilde{R}_0, \tilde{R}_0+\tilde{R}_1)\in \tilde{\mathcal{R}}(D_1,D_2)$. This implies there exists a $(2^{n\tilde{R}_0},2^{n\tilde{R}_1},n)$ code (cf. Definition \ref{definition2}), for the setting of successive refinement with cribbing (Fig. \ref{sr_cribbing}), such that distortion constraints $D_1+\epsilon$ and $D_2+\epsilon$ are met at the decoders.  We can use this code to generate a code of rates $R_1=\tilde{R}_0+\tilde{R}_1$ and $R_{12}=\tilde{R}_0$, for our cascade setting with cribbing (Fig. \ref{cascade_cribbing}), with exactly same distortions at the decoders. This proves one direction. 
\par
For the other direction, i.e., $\mathcal{R}_{cascade}(D_1,D_2)\subseteq \tilde{\mathcal{R}}(D_1,D_2)$, assume, $(R_{12},R_1)\in \mathcal{R}_{cascade}(D_1,D_2)$, which means there exist codes with rates $(R_{12},R_1)$ with decoders incurring distortions, $D_1+\epsilon$ and $D_2+\epsilon$. Assume the messages sent on first and second link in our cascade problem be $T_{1}$ and $T_{12}$ respectively. $T_{12}$ is a function of $T_1$, and we have
\bea
nR_{12}&\ge&H(T_{12})\label{eqc1}\\
nR_1&\ge&H(T_1)=H(T_1,T_{12})=H(T_{12})+H(T_{1}|T_{12})\label{eqc2}.
\eea
Using this code, we now will construct a code for successive refinement setting with cribbing decoders. Specifically, we consider encoding in $B$ blocks where each block is of length $n$.  Denote by $T_{12}(i)$ and $T_{1}(i)$ the messages which are transmitted in the cascade source coding setting in $i^{th}$ block. Note that the tuple $\{T_{12}(1),\cdots, T_{12}(B)\}$ can be communicated to both decoders in the successive refinement setting with vanishing probability of error, with a rate $R_0=\frac{1}{n}H(T_{12})$, with large number of blocks, and similarly the tuple $(T_1(1),\cdots,T_1(B))$ can be communicated to Decoder 1 with rate (using Slepian Wolf Coding as $(T_{12}(i), T_1(i))$ are independent) $R_1=\frac{1}{n}H(T_1|T_{12})$. Thus Decoder 1 and Decoder 2 will know exactly the same $(T_{12},T_1)$ and $T_{12}$ respectively as they would know in cascade  setting. Since the cribbing structure (Decoder 2 gets the crib from Decoder 1 non-causally, strictly-causally and causally) is same in cascade source coding  and successive refinement setting, decoders will be able to achieve same distortion levels, $(D_1,D_2)$. This implies, $(\frac{1}{n}H(T_{12}), \frac{1}{n}H(T_1|T_{12}))\in\mathcal{R}(D_1,D_2)$ or $(\frac{1}{n}H(T_{12}), \frac{1}{n}H(T_1))\in\tilde{\mathcal{R}}(D_1,D_2)$, which implies by Eq. (\ref{eqc1})-(\ref{eqc2}), that $(R_0, R_1)\in\tilde{\mathcal{R}}(D_1,D_2)$.
\section{Proof of Achievability in Theorem \ref{theorem2}}
\label{appendixC}
We describe in detail the achievablility in Theorem \ref{theorem2}.
\begin{itemize}
\item \textit{Codebook Generation} : Fix the distribution
$P_{X}P_{\hat{Z}_1,\hat{X}_2|X}P_{\hat{X}_1|X,\hat{Z}_1,\hat{X}_2}$, 
$\epsilon > 0$ such that
$E[d_1(X,\hat{X}_1)]\le \frac{D_1}{1+\epsilon}$ and $E[d_2(X,\hat{X}_2)]\le
\frac{D_2}{1+\epsilon}$. Generate codebook
$\mathcal{C}_{\hat{X}_2}$ consisting of $2^{nI(X;\hat{X}_2)}$
$\hat{X}^n_2(m_h)$ codewords
generated i.i.d $\sim P_{\hat{X}_2}$, $m_h\in[1:2^{nI(X;\hat{X}_2)}]$. For each
$m_h$, generate a
codebook  $\mathcal{C}_{\hat{Z}_1}(m_h)$ consisting
of $2^{nI(X;\hat{Z}_1|\hat{X}_2)}$ $\hat{Z}^n_1$
codewords
generated i.i.d. $\sim
P_{\hat{Z}_1|\hat{X}_2}$. We then bin these generated $\hat{Z}_1^n$ codewords for each $m_h$, in $2^{nR_0}$ vertical bins, $\mathcal{B}(m_v)$, $m_v\in[1:2^{nR_0}]$ and index them accordingly with $l\in[1:2^{n(I(X;\hat{Z}_1|\hat{X}_2)-R_0)}]$. $\hat{Z}^n_1$ codewords can be indexed equivalently as the tuple $(m_h,m_v,l)$. For each $\hat{Z}^n_1(m_h,m_v,l)$ codeword, generate a codebook, 
$\mathcal{C}_{\hat{X}_1}(m_h,m_v,l)$ consisting
of $2^{nI(X;\hat{X}_1|\hat{Z}_1,\hat{X}_2)}$ $\hat{X}^n_1(m_h,m_v,l,k)$
codewords
generated i.i.d. $\sim
P_{\hat{X}_1|\hat{Z}_1,\hat{X}_2}$,
$k\in[1:2^{nI(X;\hat{X}_1|\hat{Z}_1,\hat{X}_2)}]$. Thus the generation of codebooks is similar to that in perfect cribbing, except here we generate one more layer, of $\hat{Z}_1$ codewords. Also we bin $\hat{Z}_1^n$ codewords instead of $\hat{X}_1^n$. Here, $m_h$ and $m_v$ correspond to the row and column index of the \textquotedblleft doubly-indexed\textquotedblright\ bin which contains $\hat{Z}_1^n$ codeword and for each $\hat{Z}_1^n$ codeword, a codebook of $\hat{X}_1^n$ codebook is generated.  
\item \textit{Encoding} : Given source sequence $X^n$, encoder finds the index
$m_h\in [1:2^{nI(X;\hat{X}_2)}]$ from codebook $\mathcal{C}_{\hat{X}_2}$ such that 
$(X^n,\hat{X}^n_2(m_h))\in\mathcal{T}^n_\epsilon$. The
encoder then finds the index tuple $(m_v,l)$ from the
$\mathcal{C}_{\hat{Z}_1}(m_h)$
codebook, such that $(X^n,\hat{Z}^n_1(m_h,m_v,l),\hat{X}^n_2(m_h))\in\mathcal{T}^n_\epsilon$. Encoder then finds the index $k$ from the
$\mathcal{C}_{\hat{X}_1}(m_h,m_v,l)$
codebook, such that $(X^n,\hat{X}^n_1(m_h,m_v,l,k),\hat{Z}^n_1(m_h,m_v,l),
\hat{X}^n_2(m_h))\in\mathcal{T}^n_\epsilon$. Thus $\hat{Z}^n_1\in \mathcal{B}(m_v)$. 
$m_v$ is described as $R_0$ and the index triple, $(m_h,l,k)$ is described as $R_1$, thus
\bea
R_1&\ge&I(X;\hat{X}_2)+I(X;\hat{Z}_1|\hat{X}_2)-R_0+I(X;\hat{X}_1|\hat{Z}_1,\hat{X}_2)\nonumber\\
\mbox{or, } \ \ R_0+R_1&\ge&I(X;\hat{Z}_1,\hat{X}_1,\hat{X}_2)=I(X;\hat{X}_1,\hat{X}_2)\label{eq2},
\eea 
as $\hat{Z}_1=g(\hat{X}_1)$.
\item \textit{Decoding} : Using the indices sent by encoder, Decoder 1 constructs $\hat{X}^n_1=\hat{X}_1^n(m_h,m_v,l,k)$. Decoder
2 gets $\hat{Z}^n_1$ and column index $m_v$, and infers the
unique index $m_h$ such that
$\hat{Z}^n_1=\hat{Z}^n_1(m_h,m_v,\tilde{l})$ for some $\tilde{l}\in[1:2^{n(I(X;\hat{Z}_1|\hat{X}_2)-R_0)}]$. 
\item \textit{Distortion Analysis} : Consider the following events : 
\begin{itemize}
\item 
\bea
\mathcal{E}_0&=&\mbox{Encoder cannot find $(\hat{X}_2^n,\hat{Z}_1^n,\hat{X}_1^n)$ jointly typical with given source $X^n$}  
\eea
 But the probability of this event vanishes by \textit{Covering Lemma}, Lemma \ref{covering} as there are $2^{n I(X;\hat{X}_2)}$ $\hat{X}_2^n$ codewords, for each $\hat{X}_2^n$ codeword there are $2^{nI(X;\hat{Z}_1|\hat{X}_2)}$ $\hat{Z}_1^n$ codewords and finally for each $\hat{Z}_1^n$ codeword there are $2^{n I(X;\hat{X}_1|\hat{Z}_1,\hat{X}_2)}$ $\hat{X}_1^n$ codewords . Without loss of generality, now suppose that $(m_h,m_v,l,k)=(1,1,1,1)$ was sent by the encoder.   
  \item
\bea
\mathcal{E}_1&=&\mbox{$\hat{Z}_1^n$ does not lie in bin with row index $m_h=1$ and column index $m_v=1$ }\\
&=&\bigg{\{}\hat{Z}^n_1\neq\hat{Z}^n_1(1,1,\tilde{l}), \mbox{ for any } \tilde{l}\in[1:2^{n(I(X;\hat{Z}_1|\hat{X}_2)-R_0)}]\bigg{\}}.
\eea
But the probability of this event goes to zero, because of our encoding procedure, as $\hat{Z}_1^n=
\hat{Z}_1^n(1,1,1)$.
\item
\bea
\mathcal{E}_2&=&\mbox{$\hat{Z}_1^n$ lies in bin with row index $\hat{m}_h\neq1$ and column index $m_v=1$.}\\\
&=&\bigg{\{}\hat{Z}^n_1=\hat{Z}^n_1(\hat{m}_h,1,\tilde{l}), \hat{m}_h\neq 1, \mbox{ for some } \tilde{l}\in[1:2^{n(I(X;\hat{Z}_1|\hat{X}_2)-R_0)}]\bigg{\}}.\nonumber\\
\eea
Using similar argument as in the case of perfect cribbing, probability of this event goes to zero with large $n$, if
\bea
\label{eqd.1}
I(X;\hat{X}_1,\hat{X}_2)-R_0\le I(\hat{Z}_1;\hat{Z}_1,\hat{X}_2)=H(\hat{Z}_1).
\eea
\end{itemize}
Thus consider the event, $\mathcal{E}=\mathcal{E}_0\cup\mathcal{E}_1\cup\mathcal{E}_2$, using Eq. (\ref{eq2}) and Eq. (\ref{eqd.1}), probability of this event goes to zero with large $n$ if,
\bea
R_0+R_1&\ge& I(X;\hat{X}_1,\hat{X}_2)\\
R_0&\ge& \{I(X;\hat{Z}_1,\hat{X}_2)- H(\hat{Z}_1)\}^+.
\eea
Distortion is bounded as in other sections.
\end{itemize}

\section{Proof of Theorem \ref{t_mac}, MAC with cribbing encoders and common message}
{\it Proof of achievability of Theorem \ref{t_mac}, noncausal case:} The main idea of the
achievability proof is to split message $m_1$ into two parts $m_1'$ and $m_1''$ with rates $R_1'$
and $R_1''$ respectively, such that $R_1=R_1'+R_1''$. Message $m_1'$ is transmitted to Encoder 1
through the cribbing signal $Z_1^n$, while $m_1''$ remains as a private message to Encoder 1.

{\it Code design:} For the given joint distribution $P(x_1,x_2)$ generate $2^{nR_1'}$ codewords
$z_1^n$ distributed i.i.d. according to $P(z_1)$. For each codeword $z_1^n$ generate $2^{nR_0}$
codewords $x_2^n$ according to $P(x_2|z_1).$ For each codewords pair $(z_1^n,x_2^n)$ generate
$2^{nR_1''}$ $x^n_1$ codewords according to $P(x_1|z_1,x_2)$.

{\it Encoding and decoding:}
\begin{itemize}
\item {Encoder 1:} maps $(m_1',m_1'',m_0)$ to $(z_1^n(m_1'), x_2^n(z_1^n,m_0), x_1^n(x_2^n,z_1^n,
m_1''))$, and transmits $x_1^n(x_2^n,z_1^n, m_1'')$.

\item {Encoder 2:}   transmits
$x_2^n(z_1^n,m_0)$.

\item {Decoder:} looks for $(\hat m_0, \hat m_1', \hat m_1'')$ such that

\begin{equation}
(z_1^n(\hat m_1'), x_2^n(z_1^n,\hat m_0), x_1^n(x_2^n,z_1^n, \hat m_1''), y^n)\in T_{\epsilon}^{(n)}.
\end{equation}
\end{itemize}

{\it Error analysis:} Without loss of generality let's assume that the message that is sent is
$m_0=1, m_1'=1, $ and $m_1''=1$.
\begin{itemize}
\item Let $E_0$ be the event that $(x_1^n(1),x_2^n(1),y^n) \notin T_\epsilon^{(n)}.$ Clearly, $\Pr\{E_0\}\to 0$
by the law of large numbers. Hence, for the rest of the events we can assume that
$(x_1^n(1),x_2^n(1)) \in T_\epsilon^{(n)}.$

\item Let $E_{1,j}$ be the event that $z_1^n(1)=z_1^n(j)$. And let $E_1$ be the event that there
exists an $j\neq 1$ such that $z_1^n(1)=z_1^n(j)$. Following from the definition $E_1=\cup_{j\geq
1} E_{1,j}$. Let's bound the probability of $E_1$ using the union bound and the fact that
$\Pr\{E_{b,j}\}\leq 2^{-n(H(Z_1)-\epsilon)}$.
\begin{eqnarray}
\Pr\{E_1\}&=&\Pr\{\cup_{j\geq 1} E_{1,j}\} \\
&\leq&\sum_{i\geq 2} \Pr\{E_{1,j}\} \\
&\leq&\sum_{i\geq 2} 2^{-n(H(Z_1)-\epsilon)} \\
&=& 2^{n(R_1'-H(Z_1)+\epsilon)},
\end{eqnarray}
hence, if \begin{equation}\label{e1_noncausal}R_1'<H(Z_1), \end{equation}
 $\Pr\{E_1\}\to 0$ as $n\to \infty$.

\item Let $E_{i,j,k}$ be the event probability that for $\hat m_1'=i, \hat m_0=j$, and $\hat m_1''=k$
\begin{equation}
(z_1^n(\hat m_1'), x_2^n(z_1^n,\hat m_0), x_1^n(x_2^n,z_1^n, \hat m_1''), y^n)\in T_{\epsilon}^{(n)}.
\end{equation}
Let $E_3$ be the event that exists an $(i,j,k)\neq (1,1,1)$ such that  $E_{i,j,k}$ occurs.
\begin{equation}\label{e_noncausal_allp}
\Pr\{E_3\}\leq \Pr\{\bigcup_{i\geq 2, j\geq 1, k\geq 1} E_{i,j,k}\}+\Pr\{\bigcup_{i=1, j\geq 2,
k\geq 1} E_{i,j,k}\}+\Pr\{\bigcup_{i=1, j=1, k\geq 2} E_{i,j,k}\}.
\end{equation}
Now let's bound each term. Consider the first term in the RHS of (\ref{e_noncausal_allp})
\begin{eqnarray}
\Pr\{\bigcup_{i\geq 2, j\geq 1, k\geq 1}E_{i,j,k} \}&\leq& \sum_{i=2, j=1, k=1 }^{2^{nR_1'},
2^{nR_0}, 2^{nR_1''}}
2^{-n(I(Z_1,X_1,X_2;Y)-\epsilon)}\nonumber \\
&\leq& 2^{n(R_0+R_1'+R_1''-I(Z_1,X_1,X_2;Y)+\epsilon)},
\end{eqnarray}
hence if \begin{equation}\label{e2_noncausal} R_0+R_1<I(X_1,X_2;Y),\end{equation}
 then the probability above goes to
zero. Consider the second term in the RHS of (\ref{e_noncausal_allp})
\begin{eqnarray}
\Pr\{\bigcup_{i=1, j\geq 2, k\geq 1} E_{i,j,k}\}&\leq& \sum_{ j=2, k=1 }^{ 2^{nR_0}, 2^{nR_1''}}
2^{-n(I(X_1,X_2;Y|Z_1)-\epsilon)}\nonumber \\
&\leq& 2^{n(R_0+R_1''-I(X_1,X_2;Y|Z_1)+\epsilon)},
\end{eqnarray}
hence if \begin{equation}\label{e3_noncausal} R_0+R_1''<I(X_1,X_2;Y|Z_1), \end{equation}
 then the probability above goes
to zero.

Consider the third term in the RHS of (\ref{e_noncausal_allp})
\begin{eqnarray}
\Pr\{\bigcup_{i=1, j=1 , k\geq 2} E_{i,j,k}\}&\leq& 2^{n(R_1''-I(X_1;Y|Z_1,X_2)+\epsilon)},
\end{eqnarray}
hence if \begin{equation}\label{e4_noncausal} R_1''<I(X_1;Y|Z_1,X_2), \end{equation}
 then the probability above goes
to zero.

Gathering (\ref{e1_noncausal}), (\ref{e2_noncausal}), (\ref{e3_noncausal}) and
(\ref{e4_noncausal}) we obtain
\begin{eqnarray}
R_1'&<&H(Z_1)\\
R_0+R_1&<&I(X_1,X_2;Y)\\
R_0+R_1''&<&I(X_1,X_2;Y|Z_1)\\
R_1''&<&I(X_1;Y|Z_1,X_2).
\end{eqnarray}
Using Fourier$-$Motzkin elimination \cite{LecturesConvexSets10} we obtain
\begin{eqnarray}\label{e_fourier}
R_0+R_1&<&I(X_1,X_2;Y)\\
R_0+R_1&<&I(X_1,X_2;Y|Z_1)+H(Z_1) \\
R_1&<&I(X_1;Y|Z_1,X_2)+H(Z_1).
\end{eqnarray}
Since $I(X_1,X_2;Y)\leq I(X_1,X_2;Y|Z_1)+H(Z_1)$ the second inequality in (\ref{e_fourier}) is
redundant and therefore the region
\begin{eqnarray}\label{e_fourier}
R_0+R_1&<&I(X_1,X_2;Y)\nonumber \\
R_1&<&I(X_1;Y|Z_1,X_2)+H(Z_1).
\end{eqnarray}
is achievable.\hfill \QED
\end{itemize}

{\it Proof of converse for the non causal case:} Let ($2^{nR_0},2^{nR_1},n$) be a non
causal cribbing MAC code as defined in Def. \ref{def_mac_code} with a probability of error
$P_e^{(n)}$. Consider,
\begin{eqnarray}\label{e_conv_noncausal1}
R_0+R_1&=& H(M_0,M_1) \\
&=& I(M_0,M_1;Y^n)+H(M_0,M_1|Y^n) \\
&\stackrel{(a)}{\leq}& I(X_1^n,X_2^n;Y^n)+n\epsilon_n\\
&\stackrel{(b)}{\leq}& \sum_{i=1}^n I(X_{1,i},X_{2,i};Y_i)+n\epsilon_n \\
&\stackrel{(c)}{\leq}& n I(X_{1,Q},X_{2,Q};Y_Q|Q)+n\epsilon_n \\
 &\stackrel{}{\leq}& n I(X_{1,Q},X_{2,Q};Y_Q)+n\epsilon_n,
\end{eqnarray}
where (a) follows from Fano's inequality where $\epsilon_n=(\frac{1}{n}+R_0+R_1)P_e^{(n)}$, step
(b) follows from the memoryless nature of the MAC  and (c) follows from denoting $Q$ as uniform random
variable over the alphabet $\{1,2,...,n\}$. Now consider
\begin{eqnarray}\label{e_conv_noncausal2}
R_1&=& H(M_1) \\
&=& H(M_1|M_0) \\
&\stackrel{(a)}{\leq}& I(M_1;Y^n|M_0)+n\epsilon_n\\
&\stackrel{}{\leq}& I(X_1^n,Z_1^n;Y^n|X_2^n)+n\epsilon_n \\
&\stackrel{}{=}& I(Z_1^n;Y^n|X_2^n)+I(X_1^n;Y^n|X_2^n,Z_1^n)+n\epsilon_n \\
&\stackrel{(b)}{\leq}& \sum_{i=1}^n H(Z_{1,i})+I(X_{1,i};Y_i|X_{2,i},Z_{1,i})+n\epsilon_n \\
&\stackrel{(c)}{\leq}& \sum_{i=1}^n H(Z_{1,Q}|Q)+I(X_{1,Q};Y_Q|X_{2,Q},Z_{1,Q},Q)+n\epsilon_n \\
&\stackrel{}{\leq}& \sum_{i=1}^n H(Z_{1,Q})+I(X_{1,Q};Y_Q|X_{2,Q},Z_{1,Q})+n\epsilon_n,
\end{eqnarray}
where the justification for (a), (b) and (c) follows from similar arguments as steps (a), (b) and
(c) for bounding $R_0+R_1$. Since the rate pair is achievable,  the code type is arbitrary
close to the restricted distribution $P(x_1,x_2)$ and using Lemma \ref{l_XqYqZq_type} we conclude
that the distribution of $X_{1,Q}$, $X_{2,Q}$ is arbitrary close to the restricted distribution
$P(x_1,x_2)$. Finally, by denoting $Z_1=Z_Q$, $X_1=X_{1,Q}$, $X_2=X_{2,Q}$ and $Y=Y_{Q}$ and
taking into account that $P_e^{(n)}$ is going to zero as $n\to \infty$ we obtain that the region
$\mathcal R^{nc}(P)$ upper bound the capacity region.\hfill \QED

{\it Proof of achievability of Theorem \ref{t_mac}, strictly causal case:} The main idea of the
achievability proof is to combine the rate splitting idea that we used in the noncausal case with
the Markov block coding. We assume that the transmission is done in a block of size $nB$ where
$B$ is the number of subblocks and each subblock is of length $n$. Let $m_{0,b}, m_{1,b}$ be the
messages sent in block $b$. Similarly to the noncausal case, split message $m_{1,b}$ into two
parts $m_{1,b}'$ and $m_{1,b}''$ with rates $R_1'$ and $R_1''$ respectively, such that
$R_1=R_1'+R_1''$. Message $m_{1,b}'$ is transmitted to Encoder 1 through the cribbing signal,
while $m_{1,b}''$ remains as a private message to Encoder 1. Because of the causality, the
 the message $m_{1,b}'$ is known to Encoder 2 only at the end of block $b$. 

{\it Code design:} For fixed a joint distribution $P(x_1,x_2)$ generate $2^{n(R_0+R_1')}$
codewords $x_2^n$ each associated with the pair of messages $(m_{0,b},m_{1,b-1}')$. For each
codeword $x_2^n$ generate $2^{nR_1'}$ codewords $z_1^n$  according to conditional distribution
$P(z_1|x_2)$ associated with $m_{1,b}'$. For each codeword pair $(z_1^n,x_2^n)$ generate
$2^{nR_1''}$ codewords $x_1^n$  according to conditional distribution $P(x_1|z_1,x_2)$ associated
with $m_{1,b}''$.

{\it Encoding and decoding:}
\begin{itemize}
\item {Encoder 1:} In block $b$ maps $(m_{1,b-1}',m_{1,b}',m_{1,b}'',m_{0,b})$ to
$(x_2^n(m_{0,b},m_{1,b-1}'),z_1^n(m_{1,b}',x_2^n), x_1^n(m_{1,b}'',x_2^n,z_1^n))$, and transmits
$x_1^n(m_{1,b}'',x_2^n,z_1^n))$.

\item {Encoder 2:}   Transmits
$x_2^n(m_{0,b},m_{1,b-1}')$. Message $m_{1,b-1}'$ is known to Encoder 2 since at the end of block
$b-1$, $z_1^n(m_{1,b-1}',x_2^n)$ and $x_2^n$ are known.

\item {Decoder:} Does backward decoding. We assume that when decoding block $b$ message $m_{1,b}'$ is known and it
looks for tuple ($\hat m_{0,b}, \hat m_{1,b-1}', \hat m_{1,b}''$) such that

\begin{equation}
(x_2^n(\hat m_{0,b},\hat m_{1,b-1}'),z_1^n(m_{1,b}',x_2^n), x_1^n(\hat m_{1,b}'',x_2^n,z_1^n),y^n)\in
T_{\epsilon}^{(n)}.
\end{equation}
\end{itemize}

{\it Error analysis:} Without loss of generality let's assume that the message that is sent is
$m_{0,b}=1, m_{1,b}'=1, m_{1,b-1}'=1, $ and $m_{1,b}''=1$.
\begin{itemize}
\item Let $E_0$ be the event that $(x_1^n(1),x_2^n(1)) \notin T_\epsilon^{(n)}.$ Clearly, $\Pr\{E_0\}\to 0$
by the law of large numbers. Hence, for the rest of the events we can assume that
$(x_1^n(1),x_2^n(1)) \in T_\epsilon^{(n)}.$

\item Let $E_{1}$ be the event that in block $b-1$ there
exists an $j\neq 1$, such that  $z_1^n(1)=z_1^n(j)$ for some codeword $x_2^n$. Similar to the
analysis for the noncausal case
\begin{eqnarray}
\Pr\{E_1\} &=& 2^{n(R_1'-H(Z_1|X_2)+\epsilon)},
\end{eqnarray}
hence, if \begin{equation}\label{e1_strc_causal}R_1'<H(Z_1|X_2), \end{equation}
 $\Pr\{E_1\}\to 0$ as $n\to \infty$.

\item Let $E_{i,j,k}$ be the event probability that for $\hat m_{1,b-1}'=i, \hat m_{0,b}=j$, and $\hat
m_{1,b}''=k$, given that $m_{1,b}'$ is known correctly from pervious subblock decoding:
\begin{equation}
(x_2^n(\hat m_{0,b},\hat m_{1,b-1}'),z_1^n(m_{1,b}',x_2^n), x_1^n(\hat m_{1,b}'',x_2^n,z_1^n),y^n)\in
T_{\epsilon}^{(n)}.
\end{equation}
Let $E_3$ be the event that exists an $(i,j,k)\neq (1,1,1)$ such that  $E_{i,j,k}$ occurs.
\begin{equation}\label{e_str_causal_allp}
\Pr\{E_3\}\leq \Pr\{\bigcup_{(i,j)\neq(1,1), k\geq 1} E_{i,j,k}\}+\Pr\{\bigcup_{(i,j)=(1,1),
k\geq 2} E_{i,j,k}\}.
\end{equation}
Now let's bound each term. Consider the first term in the RHS of (\ref{e_str_causal_allp})
\begin{eqnarray}
\Pr\{\bigcup_{(i,j)\neq(1,1), k\geq 1} E_{i,j,k}\}&\leq& 2^{n(R_0+R_1-I(Z_1,X_1,X_2;Y)+\epsilon)},
\end{eqnarray}
hence if \begin{equation}\label{e2_strc_causal} R_0+R_1<I(X_1,X_2;Y)\end{equation}
 then the probability above goes to
zero. Consider the second term in the RHS of (\ref{e_str_causal_allp})
\begin{eqnarray}
\Pr\{\bigcup_{(i,j)=(1,1), k\geq 2} E_{i,j,k}\}&\leq& 2^{n(R_1''-I(X_1;Y|Z_1,X_2)+\epsilon)},
\end{eqnarray}
hence if \begin{equation}\label{e3_strc_causal} R_1''<I(X_1;Y|Z_1,X_2), \end{equation}
 then the probability above goes
to zero.

Gathering (\ref{e1_strc_causal}), (\ref{e2_strc_causal}), and (\ref{e3_strc_causal}) we obtain
\begin{eqnarray}
R_1'&<&H(Z_1|X_2)\\
R_0+R_1&<&I(X_1,X_2;Y)\\
R_1''&<&I(X_1;Y|Z_1,X_2).
\end{eqnarray}
Using Fourier$-$Motzkin elimination
\begin{eqnarray}\label{e_fourier}
R_0+R_1&<&I(X_1,X_2;Y)\\
R_1&<&I(X_1;Y|Z_1,X_2)+H(Z_1|X_2).
\end{eqnarray}
is achievable.\hfill \QED
\end{itemize}

{\it Proof of converse for the strictly causal case:} Let ($2^{nR_1},2^{nR_0},n$) be a 
strictly causal cribbing MAC code as defined in Def. \ref{def_mac_code} with a probability of
error $P_e^{(n)}$. Following the exact same steps as in the converse of the noncausal case in
(\ref{e_conv_noncausal1}) we obtain
\begin{eqnarray}\label{e_conv_noncausal1}
R_0+R_1 &\stackrel{}{\leq}& n I(X_{1,Q},X_{2,Q};Y_Q)+n\epsilon_n.
\end{eqnarray}
Following the exact same first four steps as in converse of the non causal case to bound $R_1$, 
(\ref{e_conv_noncausal2}) we obtain
\begin{eqnarray}
R_1
&\stackrel{}{\leq}& I(Z_1^n;Y^n|X_2^n)+I(X_1^n;Y^n|X_2^n,Z_1^n)+n\epsilon_n \\
&\stackrel{}{\leq}& \sum_{i=1}^n H(Z_{1,i}|X_{2,i})+I(X_{1,i};Y_i|X_{2,i},Z_{1,i})+n\epsilon_n \\
&\stackrel{}{\leq}& \sum_{i=1}^n H(Z_{1,Q}|X_{2,Q})+I(X_{1,Q};Y_Q|X_{2,Q},Z_{1,Q})+n\epsilon_n,
\end{eqnarray}
Since the rate pair is achievable,  the code type is arbitrary close to the restricted
distribution $P(x_1,x_2)$ and using Lemma \ref{l_XqYqZq_type} we conclude that the distribution
of $X_{1,Q}$, $X_{2,Q}$ is arbitrary close to the restricted distribution $P(x_1,x_2)$. Finally,
by denoting $Z_1=Z_Q$, $X_1=X_{1,Q}$, $X_2=X_{2,Q}$ and $Y=Y_{Q}$ and taking into account that
$P_e^{(n)}$ is going to zero as $n\to \infty$ we obtain that the region $\mathcal R^{nc}(P)$
upper bound the capacity region.\hfill \QED

{\it Proof of achievability of Theorem \ref{t_mac}, causal case:} In this proof we show how the
causal case achievability follows directly from the proof of the strictly causal case with one
modification: instead of codewords $x_2^n$ we generate codewords $u^n$, and the input to the
channel is $x_{2,i}=f(u_i,x_{1,i})$. This is possible since Encoder 2 observes causally the
signal from Encoder 1. By replacing $X_2$ with $U$ in  $\mathcal R^{sc}(P)$ and applying
$x_{2,i}=f(u_i,z_{1,i})$ and taking into account the equality
$I(Y;X_1,U)=I((Y;X_1,U,f(Z_1,U))=I(Y;X_1,X_2)$ we obtain  the region $\mathcal R^{c}(P)$
.\hfill \QED

{\it Proof of converse for the  causal case:} Let ($2^{nR_0},2^{nR_1},n$) be a partial strictly
causal cribbing MAC code as defined in Def. \ref{def_mac_code} with a probability of error
$P_e^{(n)}$. Following the exact same steps as in the converse of the noncausal case in
(\ref{e_conv_noncausal1}) we obtain
\begin{eqnarray}\label{e_conv_strcausal1}
R_0+R_1 &\stackrel{}{\leq}& n I(X_{1,Q},X_{2,Q};Y_Q)+n\epsilon_n,
\end{eqnarray}
Now consider
\begin{eqnarray}\label{e_conv_strcausal2}
R_1&=& H(M_1) \\
&=& H(M_1|M_0) \\
&\stackrel{(a)}{\leq}& I(M_1;Y^n|M_0)+n\epsilon_n \\
&\stackrel{}{\leq}& I(X_1^n,Z_1^n;Y^n|M_0)+n\epsilon_n \\
&\stackrel{}{=}& I(Z_1^n;Y^n|M_0)+I(X_1^n;Y^n|M_0,Z_1^n)+n\epsilon_n\\\
&\stackrel{(b)}{\leq}& \sum_{i=1}^n I(Z_{1,i};Y^n|M_0,Z^{i-1})+I(X_{1,i};Y_{i}|X_{2,i},M_0,Z_1^n,X^{i-1})+n\epsilon_n \\
&\stackrel{(c)}{\leq}& \sum_{i=1}^n H(Z_{1,i}|M_0,Z^{i-1})+I(X_{1,i};Y_{i}|X_{2,i},M_0,Z_1^{i-1})+n\epsilon_n\ \\
&\stackrel{}{\leq}& \sum_{i=1}^n H(Z_{1,Q}|U_Q)+I(X_{1,Q};Y_Q|X_{2,Q},U_{Q})+n\epsilon_n,
\end{eqnarray}
where (a) follows from Fano's inequality where $\epsilon_n=(\frac{1}{n}+R_1)P_e^{(n)}$, step (b)
follows from the memoryless of the MAC  and (c) follows from denoting $U_i\triangleq
(M_0,Z_1^{i-1})$ and $Q$ as uniform random variable over the alphabet $\{1,2,...,n\}$. Note that
indeed $X_{2,i}=f(M_0,Z_1^{i-1},X_{1,i})$ and therefore $X_{2,Q}=f(U_Q,X_{2,Q}).$ Rest of the steps for the completion of proof follow similar arguments as in non causal and strictly causal case. \hfill \QED

\end{document}